\documentclass[12pt]{article}

\usepackage{amsmath,amsfonts,latexsym,amssymb,amsthm,graphicx}

\usepackage{bm}
\usepackage[round]{natbib}
\usepackage{hyperref}
\usepackage[margin=.8in]{geometry}
\usepackage{setspace}
\usepackage[normalem]{ulem}
\usepackage{subcaption}
\usepackage{multirow}
\usepackage{xcolor}
\usepackage{enumitem}

\newtheorem{theorem}{Theorem}
\newtheorem{lemma}{Lemma}

\newtheorem{corollary}{Corollary}

\theoremstyle{definition}

\def\hatxi{{\hat{\xi}}}
\def\hatEps{{\hat{\mathcal{E}}}}

\def\hatEps{\hat{\mathcal{E}}}
\def\Eps{{\mathcal{E}}}
\def\E{{\mathrm{E}}}

\def\cov{{\mathrm{cov}}}
 
\def\bX{{\boldsymbol{X}}}
\def\bU{{\boldsymbol{U}}}
\def\bbU{{\mathbf{U}}}
\def\bu{{\mathbf{u}}}

\def\bTheta{{\boldsymbol{\Theta}}}
\def\bOmega{{\boldsymbol{\Omega}}}
\def\btheta{{\boldsymbol{\theta}}}

\def\ba{{\mathbf{a}}}
\def\bh{{\mathbf{h}}}
\def\bM{{\mathbf{M}}}
\def\bEta{{\boldsymbol{\eta}}}
\def\bS{{\mathbf{S}}}
\def\bs{{\mathbf{s}}}
\def\bw{{\mathbf{w}}}
\def\bJ{{\mathbf{J}}}

\def\cE{{\mathcal{E}}}

\title{High-dimensional covariance estimation by pairwise likelihood truncation}

\author{Alessandro Casa$^1$,  Davide Ferrari$^2$ and Zhendong Huang$^3$ \\
{\small $^1$Departments of Economics, Università degli Studi di Bergamo} \\
{\small $^2$Faculty of Economics and Management, Free University of Bozen-Bolzano} \\
{\small $^3$School of Mathematics and Statistics, The University of Melbourne}}
\date{}
\begin{document}

\maketitle

\begin{abstract}

Pairwise likelihood is a useful approximation to the full likelihood function for covariance estimation in high-dimensional context. It simplifies high-dimensional dependencies by combining marginal bivariate likelihood objects, thus making estimation more manageable. In certain models, including the Gaussian model, both pairwise and full likelihoods are maximized by the same parameter values, thus retaining optimal statistical efficiency, when the number of  variables is fixed. Leveraging on this insight, we introduce estimation of sparse high-dimensional covariance matrices by maximizing a truncated version of the pairwise likelihood function, obtained by including pairwise terms corresponding to nonzero covariance elements. To achieve a meaningful truncation, we propose to minimize the $L_2$-distance between pairwise and full likelihood scores plus an $L_1$-penalty discouraging the inclusion of uninformative terms. Differently from other regularization approaches, our method focuses on selecting whole pairwise likelihood objects rather than shrinking individual covariance parameters, thus retaining the inherent unbiasedness of the pairwise likelihood estimating equations. This selection procedure is shown to have the selection consistency property as the covariance dimension increases exponentially fast. Consequently, the implied pairwise likelihood estimator is consistent and converges to the oracle maximum likelihood estimator assuming knowledge of nonzero covariance entries.
\end{abstract}




\section{Introduction}\label{sec:introduction}

Covariance estimation is a fundamental problem in multivariate analysis and has a range of applications in many scientific fields.  In genetics and genomics, it helps our understanding of complex traits by identifying genetic correlations. In finance,  it is essential for risk management and asset allocation strategies. In environmental science and climate studies, it helps us represent interactions between climate variables and predicting future trends. In statistics, covariance estimation is a cornerstone for more complex dependence models such as principal component analysis, factor analysis, and linear models. As new technology emerges, covariance estimation becomes increasingly challenging due to the growth of the number of variables involved, leading traditional methods to become more computationally intensive and susceptible to statistical error accumulation.

Let $\bX = (X_1, \dots, X_p)^\top$ be a $p$-variate random vector with zero mean and  covariance $\cov(\bX) = \bTheta=(\theta_{jk})_{p\times p}$. We use $\btheta = \text{vech}(\bTheta)$ to denote the $m \times 1$ column vector obtained by vectorizing the upper triangular part of $\bTheta$, where $m = p(p+1)/2$ is the total number of variance and covariance parameters. Given an i.i.d. random sample   $\bX^{(1)}, \dots,  \bX^{(n)}$ from the distribution of $\bX$, we wish to estimate the covariance matrix $\bTheta$.  If $\bX$ follows a multivariate Gaussian distribution, the log-likelihood function ignoring any constant term is
$
\ell(\btheta) =   -\log |\bTheta| -  \text{tr}(\bTheta^{-1}\bS)
$. The sample covariance matrix
\begin{equation}\label{eq:sample_cov}
\bS = (S_{jk})_{p\times p} := \dfrac{1}{n} \sum_{i=1}^n \bX^{(i)} \bX^{(i)^\top}
\end{equation} 
maximizes $\ell(\btheta)$ and is the ML estimator of $\bTheta$.  When the data are non-Gaussian or weakly dependent, $\ell(\btheta)$ is interpreted as a quasi-likelihood function, yet it remains a valid objective function for estimation. The high-dimensional setting where the dimension $p$ is larger than the sample size $n$ is of particular interest in current applications, but poses nontrival challenges. In this setting, the ML estimator, without additional restrictions on the  covariance structure,  is singular and inconsistent
due to the accumulation of excessive estimation error and too many free parameters \citep{bai1993limit, bai1988note}. To deal with the overwhelming number of parameters,  a common assumption in the literature also adopted here is that  $\bTheta$ is sparse, meaning that a number of its entries are exactly zero. Particularly, we define the index set $\cE=\{ jk : j \le  k, \theta_{jk}\neq 0 \}$ corresponding to diagonal elements of $\bTheta$ and the pairs for which covariances are actually different from zero. We use $m_0 = \text{card}(\cE)$ to denote the cardinality of $\cE$, while $\setminus \cE = \{ jk : j < k; \theta_{jk} = 0 \}$ denotes the complement of $\cE$.

Many approaches for estimating high-dimensional covariance matrices have been proposed; e.g., see reviews in  Lam \cite{lam2020high}, Fan et al \cite{fan2016overview} and Pourahmadi \cite{pourahmadi2013high}. Given the extensive body of literature, here we focus  on penalized likelihood and thresholding methods due to their relevance to our own work. Penalized likelihood estimators achieve estimation of sparse covariance matrices by maximizing the objective function $ \ell(\btheta) -  p_\lambda(\btheta)$, where $p_\lambda(\btheta)$ is a penalty term that encourages the off-diagonal elements of $\bTheta$ to shrink towards zero \citep{bien2011sparse, chaudhuri2007estimation,lam2009sparsistency,rothman2012positive,wang2014coordinate}. Thus, a certain amount of bias in the  estimating equations is necessary to induce the desired selection effect. Theoretical guarantee for sparsistency -- the property where all zero parameters are estimated as zero with probability tending to one -- requires that the number of non-zero entries be small, $O(1)$ among $O(p^2)$ parameters, in the worst case scenario \citep{lam2009sparsistency}. Due to the connection with the Gaussian graphical models and the convenient interpretation in terms of conditional dependence, penalization of the precision matrix $\bOmega = \bTheta^{-1}$ is also frequently considered \citep{banerjee2008model,friedman2008sparse, meinshausenbuhlmann2006, scheinberg:2010, witten:2011, whittaker:1990, yuan2008efficient, yuan2007model}.

Another popular approach to obtain sparse estimates is thresholding \citep{bickel2008covariance, cai2012optimal}, which modifies the sample covariance matrix by setting a portion of the estimated off-diagonal elements to zero. Thresholding reduces noise accumulation by avoiding the estimation of irrelevant off-diagonal elements of $\bTheta$. Hard thresholding sets to zero the $jk$th covariance parameter if $|S_{jk}| \leq \psi$, for some threshold $\psi>0$. Adaptive thresholding \citep{cai2011adaptive} also takes into account heterogeneous estimation uncertainty; particularly, the estimated off-diagonal elements are set to zero if $|S_{jk}|/\text{SE}_{jk} \leq \psi$, $j \neq k$, where $\text{SE}_{jk}$ is the standard error of the marginal estimate $S_{jk}$. Generalized thresholding extends these approaches  by replacing hard thresholding operators with soft-thresholding operators \citep{rothman2009generalized}. Other established methods for high-dimensional covariance estimation   include regularization on the eigenstructure of the matrix \citep{daniels2001shrinkage,ledoitWolf2012}, banding and tapering \citep{bickelLevina2008banding,bienBandingJASA,bien2016convex, furrer2006covariance, xue2014rank}.

In this paper, we introduce a methodology for sparse covariance estimation in high-dimensional problems within the composite likelihood framework. A composite likelihood function is a type of approximation to a complex likelihood function obtained through the combination of several low-dimensional likelihood objects. Besag \cite{besag75} pioneered composite likelihood estimation for spatial data, while Lindsay \cite{Lindsay88} developed the methodology in its generality; see Varin and al. \cite{varin11} for a  survey.  A special type of composite likelihood useful for covariance estimation is the pairwise likelihood (PL), where the low-dimensional likelihoods are  based on bivariate marginal densities \citep{cox2004note}. The PL estimation framework is particularly attractive for covariance estimation since, for common exponential family models, the PL function is maximized at the same values as the full likelihood function \citep{mardia2009maximum}. In principle, this feature enables one to retain  full statistical  efficiency of the resulting estimator, provided that the correct model structure  is identified.

Overall, PL estimators show considerable potential for high-dimensional covariance estimation due to the simplicity in defining the objective function,  their computational advantages and the statistical efficiency in common model families when $p$ is small compared to the number of observations. However, the application of traditional PL estimators in large problems is still hindered by their poor statistical efficiency  in the $p>n$ setting. The main challenges are similar to those for ML estimation:  without some form of restriction on the model structure,  PL estimators accumulate excessive error and are not generally consistent as $p$ diverges \citep{cox2004note,ferrari2016parsimonious, huang2022fast}. More recently, building on the success of shrinkage methods for the full likelihood, numerous studies have extended the use of sparsity-inducing penalties to the composite likelihood framework, applying it to various specialized models \citep{bradic2011penalized, xue2012nonconcave, gao2017data, hui2018sparse}.

Here we adopt a completely different approach compared to existing penalized procedures for PL estimation. Instead of shrinking individual elements of the parameter vector $\btheta$ to zero, we propose to discourage the inclusion of entire pairwise likelihood objects in the final PL objective function.  To highlight this important feature, this approach will be referred to as truncated pairwise likelihood (TPL) estimation in the remainder of the paper. The main conceptual innovation lies in recognizing that these partial pairwise components are not strictly statistical parameters. Therefore, penalizing the log-pairwise likelihood function may not be the most suitable approach for eliminating such components, and instead we opt for an alternative statistical criterion.  Particularly, the theory of unbiased estimating equations suggests that one should minimize the distance between the candidate PL and the ML score functions or, equivalently,  maximize the statistical efficiency of the underlying PL estimator  \citep{heyde97, Lindsay&al11}. Based on this understanding, we propose to minimize a convex criterion related to statistical efficiency with the addition of a weighted $L_1$-penalty to discourage the selection of too many noisy or redundant pairwise scores terms.

One advantage of the TPL method is that it inherently yields unbiased estimates. While conventional penalization techniques introduce bias in the estimating equations to shrink estimates towards zero, our penalized criterion operates by selecting entire bivariate models, thereby preserving the unbiasedness from the corresponding pairwise score function, provided that the set of selected pairwise likelihoods is sufficiently comprehensive, i.e., it includes at least those covariances that are non-zero. Moreover, for Gaussian data, these estimates coincide with the Maximum Likelihood Estimates (MLE), thus yielding optimal efficiency. From a theoretical viewpoint, the procedure is shown to have the selection consistency property as $p$ grows exponentially in $n$ (Theorem \ref{thm:consistency}). A meaningful consequence is that selection consistency implies that our estimator converges in probability to the oracle ML estimator, i.e., the ML estimator that assumes knowledge of the true covariance structure.

Another significant contribution of this paper is to emphasize the subtle connection between pairwise likelihood estimation -- or, more broadly, composite likelihood estimation -- and adaptive thresholding. The proposed penalized criterion functions as an adaptive thresholding mechanism that adjusts standard errors using information from the selected pool of covariances; see Section \ref{sec:first_order}. Leveraging on joint information from multiple covariances rather than marginal information, this adaptive adjustment can be exploited to enhance the performance of standard adaptive thresholding approaches.

The paper is structured as follows. In Section \ref{sec:method}, we outline our main methodology, study the relationship with adaptive thresholding and present an efficient coordinate descent algorithm for its implementation. In Section \ref{sec:prop}, we study the theoretical properties of our method, including scenarios where the number of parameters diverges rapidly as the sample size increases.
In Section \ref{sec:num}, we study the finite sample properties of our method through simulations under both sparse and relatively dense covariance structures (Section \ref{sec:MC}), and by analyzing genetic single-cell RNA-sequencing data on colorectal cancer (Section \ref{sec:real_data}). We conclude with closing remarks in Section \ref{sec:conclusion}. Technical proofs of the main theorems and lemmas are provided in the Appendix.

\section{Pairwise likelihood truncation via $L_1$-penalization} \label{sec:method}

\subsection{An efficiency criterion for pairwise likelihood truncation} \label{sec:criterion}

In this section, we introduce an empirical criterion for selecting the relevant pairwise contributions. We adopt the PL estimation framework due to a crucial property: the PL function attains its maximum at the same parameter values as the full likelihood function for certain covariance models within the exponential family, such as the Gaussian distribution; see Mardia \citep{mardia2009maximum}. Hence, the resulting PL estimator achieves full efficiency, at least when the sparsity structure is known in advance.  Leveraging the PL framework enables us to simplify the complexity of the full likelihood while retaining a substantial amount of statistical information. Next, we formulate our statistical criterion with the multivariate Gaussian distribution in mind, which is common in multivariate analysis, but the criterion remains applicable in other scenarios as well.

Let  $\bTheta_{jk}=\cov\{ (X_j, X_k)^\top\}$ be the $2\times 2$ sub-matrix of $\bTheta$ with elements corresponding to the $j$th and $k$th columns and rows. From one vector $\bX$ we form first- and second-order log-likelihoods contributions
\begin{align}\label{eq:ellij}
\ell_{jj}(\btheta; \bX)  &=   - \log (\theta_{jj}) - \dfrac{X^2_{j}}{\theta_{jj}},  \ \  j = 1,\dots, p,  \notag \\
\ell_{jk}(\btheta; \bX) &  = - \log (\theta_{jj} \theta_{kk} - \theta_{jk}^2) - \dfrac{\theta_{kk} X^2_{j} - 2 \theta_{jk} X_{j}X_k + \theta_{jj} X^2_{k}}{\theta_{jj} \theta_{kk} - \theta_{jk}^2}, \ \  j < k,
\end{align}
based on univariate and bivariate normal distributions, respectively. Marginal and bivariate scores are obtained, as usual, by differentiation. Particularly, marginal scores are represented by the $m \times 1$ vector defined by
\begin{equation}\label{eq:score_mar}
\bu_{jj}(\btheta; \bX) = \dfrac{\partial}{\partial \btheta}\ell_{jj}(\btheta; \bX) = \left( 0,\dots,0, u_{jj} ,0,\dots,0 \right)^\top,
\end{equation}
where $u_{jj} = (X^2_{j}- \theta_{jj})/\theta^2_{jj}$ is the only non-zero element in position $(2p+2-j)(j-1)/2+1$. Bivariate scores are represented by the $m \times 1$ vector
\begin{equation}\label{eq:score_mar2}
\bu_{jk}(\btheta; \bX) = \dfrac{\partial}{\partial \btheta}\ell_{jk}(\btheta; \bX) = \left( 0,\dots,0, u^{[1]}_{jk},0,\dots,0, u^{[3]}_{jk},0,\dots,0, u^{[2]}_{jk},0,\dots,0 \right)^\top,
\end{equation}
for $j<k$, with  three nonzero elements
\begin{align}\label{eq:u_{jk}}
	 u_{jk}^{[1]} & = -\frac{ \theta_{jj} \theta_{kk}^2 - \theta_{kk}\theta_{jk}^2 - X^2_{k}\theta_{jk}^2
- X^2_{j}\theta_{kk}^2+2X_{j}X_k \theta_{jk}\theta_{kk}}{ 2(\theta_{jj}\theta_{kk}-\theta_{jk}^2)^2}, \notag \\
	 u_{jk}^{[2]} & = -\frac{ \theta_{kk} \theta_{jj}^2 - \theta_{jj}\theta_{jk}^2 - X_{j}^2\theta_{jk}^2
- X_{k}^2\theta_{jj}^2+2X_{j}X_k \theta_{jk}\theta_{jj}}{ 2(\theta_{jj}\theta_{kk}-\theta_{jk}^2)^2}, \notag \\
	 u_{jk}^{[3]} & = -\frac{ \theta_{jk} ^3 -\theta_{jk}\theta_{jj} \theta_{kk} + X_{j}^2\theta_{jk} \theta_{kk} +X_{k}^2\theta_{jk} \theta_{jj} -X_{j}X_k ( \theta_{jj}\theta_{kk}+ \theta_{jk} ^2)  }
{ (\theta_{jj}\theta_{kk}-\theta_{jk}^2)^2},
\end{align}
in positions $(2p+2-j)(j-1)/2+1$, $(2p+2-k)(k-1)/2+1$ and $(2p+2-j)(j-1)/2+k+1-j$, respectively.

We define the overall PL score function  by taking the linear combination of marginal and bivariate score contributions as follows:
\begin{equation}\label{eq:score_pw}
\bu(\btheta, \bw; \bX) = \bbU(\btheta; \bX) \bw = \sum_{j\le k} w_{jk} \bu_{jk}(\btheta; \bX),
\end{equation}
where $\bw \in \mathbb{R}^m$ is a vector of linear coefficients with elements $w_{jk}$ associated to the $jk$th score vector $\bu_{jk}$, and $\bbU(\btheta, \bX)$ is the $m \times m$ matrix with columns given by the score vectors $\{ \bu_{jk}, j\le k \}$. Equation (\ref{eq:score_pw}) coincides with the PL function of  Cox and Reid \cite{cox2004note} when $w_{jk} = 1$, for all $j< k$, and $w_{jk}=-ap$, $a \in \mathbb{R}$, for all $j=k$. In other cases, it offers a more flexible generalization, allowing heterogeneous contributions from specific pairwise likelihoods.
Note that the selection of the relevant covariance structure can be achieved by setting certain bivariate coefficients to zero. In particular,  when $X_j$ and $X_k$ have zero correlation, one should have  $w_{jk}=0$, since the corresponding score $\bu_{jk}$ does not contribute meaningful information. From this viewpoint, the model selection problem of determining the nonzero covariances $\theta_{jk}$ can be recast in terms of selecting the most informative pairwise score contributions among $p(p-1)/2$ possible contributions.

To select the pairwise score terms, we propose to select the optimal coefficients $\bw(\btheta)$ by minimizing the theoretical criterion
\begin{align}\label{eq:crit_true}
	 \underbrace{\dfrac{1}{2}\E \left\| \bu^{ML}(\btheta; \bX) - \bu(\btheta, \bw; \bX) \right\|_2^2}_\text{Statistical efficiency}  + \underbrace{\dfrac{\lambda}{n} \sum_{j<k} \frac{|w_{jk}|}{\theta_{jk}^2}}_\text{Sparsity},
\end{align}
with respect to $\bw$, where $\bu^{ML}(\btheta; \bX) = \partial \log f(\bX; \btheta) / \partial \btheta$ is the ML score and $f(\cdot; \btheta)$ is the $p$-variate density of $\bX$, which may be unknown.  Minimizing (\ref{eq:crit_true}) is interpreted as decreasing the distance between ML and PL scores, effectively maximizing statistical efficiency for a given level of sparsity. Particularly,  the weighted $L_1$-penalty  $\lambda n^{-1}\sum_{j<k} |w_{jk}|\theta^{-2}_{jk}$ discourages the inclusion of too many pairwise contributions, due to its geometric properties and several elements of $\bw$ become zero when $\lambda$ is sufficiently large. This penalty is adaptive in the sense  that for $\theta_{jk}$ closer to zero, the $jk$th pairwise score receives a larger penalty. Differently from the adaptive Lasso penalty of Zou \cite{zou2006adaptive} for sparse regression, here the penalty focuses on coefficients $w_{jk}$ associated with entire score terms, while the quantity $\theta^{-2}_{jk}$ depends on the covariance parameters but not on $\bw$, thus playing the role of a fixed weight for the $jk$th term.

Straightforward algebra shows that (\ref{eq:crit_true}) can be re-written  as convex criterion not depending on the maximum likelihood score $\bu^{ML}$; hence a  solution can be efficiently obtained using available algorithms. Particularly, ignoring an irrelevant term not depending on $\bw$, we  expand (\ref{eq:crit_true}) and define the objective function:
\begin{align}\label{eq:crit_true1_2}
d_\lambda(\bw; \btheta) =   \dfrac{1}{2}\E \left\|  \bu(\btheta, \bw; \bX) \right\|_2^2- \E[ \bu^{ML}(\btheta; \bX)^\top\bu(\btheta, \bw; \bX)] +  \dfrac{\lambda}{n} \sum_{j<k} \frac{|w_{jk}|}{\theta_{jk}^2}.
\end{align}
The above criterion can be rewritten in a more compact fashion as
\begin{align}\label{eq:crit_true1_3}
d_\lambda(\bw; \btheta) & =   \dfrac{1}{2} \bw^\top \bJ(\btheta) \bw - \bw^\top \text{diag}\{ \bJ(\btheta)\} +  \dfrac{\lambda}{n} \sum_{j<k} \frac{|w_{jk}|}{\theta_{jk}^2},
\end{align}
where $\bJ(\btheta)$ is the $m \times m$ scores covariance matrix
\begin{align}
\bJ(\btheta) &= \E\left[ \bbU(\btheta; \bX)^\top  \bbU(\btheta; \bX)\right],
\end{align}
and $\text{diag}\{ \bJ(\btheta) \}$ denotes the $m \times 1$ vector containing the diagonal elements of $\bJ(\btheta)$. Note that dependence on the ML score in the second term of (\ref{eq:crit_true1_2}) can be dropped thanks to the second Bartlett's equality which, assuming unbiased scores, implies that the diagonal elements of $\bJ(\btheta)$ are
$- \E[\bu^{ML}(\btheta; \bX)^\top  \bu_{jk}(\btheta; \bX) ] = \E \left\| \bu_{jk}(\btheta; \bX) \right\|^2_2$, for $j\le k$.

\subsection{Covariance tresholding by pairwise likelihood truncation} \label{sec:cov_thresholding} In this section, we describe our approach to obtain a sparse estimate of $\btheta$. Given an i.i.d. sample $\bX^{(1)}, \dots, \bX^{(n)}$ from the distribution of $\bX$, we solve the estimating equations $\sum_{i=1}^n \bu(\btheta, \bw; \bX^{(i)}) = \bf0$. It is easy to check that the solution corresponds to  $\theta_{jj} = S_{jj}$,  $1 \le j\le p$, and $\theta_{jk}=S_{jk}$ for those pairs $jk$ such that $w_{jk}\neq 0$. This solution is invariant with respect to the elements of $\bw$ different from zero.  Particularly, the marginal estimating equation
\begin{equation}\label{eq:est_eq}
\sum_{i=1}^n u_{jj}(\btheta;\bX^{(i)}) = \sum_{i=1}^n w_{jj}\left( \dfrac{( X_j^{(i)})^2 - \theta_{jj}}{\theta^2_{jj}} \right)= 0,
\end{equation}
is solved by $\theta_{jj}= S_{jj} = n^{-1} \sum_{i=1}^n (X_j^{(i)})^2 $, for all $1 \le j \le p$. Substituting $\theta_{jj}=S_{jj}$ and $\theta_{kk} = S_{kk}$ into $\sum_{i=1}^n \bu_{jk}(\btheta; \bX^{(i)})$ gives
\begin{align*}
	  \sum_{i=1}^nu_{jk}^{[1]} = -   \frac{   \sum_{i=1}^n X_j^{(i)}X_k^{(i)}- n\theta_{jk}   }{ {S_{kk}}^{-1} \theta_{jk}^{-1} (S_{jj} S_{kk}-\theta_{jk}^2)^2}, \quad
 \sum_{i=1}^n u_{jk}^{[2]} =  - \frac{      \sum_{i=1}^n X_j^{(i)}X_k^{(i)}- n \theta_{jk}  }{ {S_{jj}}^{-1}  \theta_{jk}^{-1} (S_{jj} S_{kk}-\theta_{jk}^2)^2},
\end{align*}
 \begin{align*}
	 \sum_{i=1}^n  u_{jk}^{[3]} = -  \frac{    \sum_{i=1}^n  X_j^{(i)}X_k^{(i)}  - n\theta_{jk}    }
{ ( \theta_{jk}^2 + S_{jj} S_{kk}  )^{-1 }(S_{jj} S_{kk}-\theta_{jk}^2)^2},
\end{align*}
where $u_{jk}^{[r]}$, $r=1,2,3$, are defined in (\ref{eq:u_{jk}}). It is easy to see that the estimating equations are solved by  $\theta_{jk}= S_{jk} = n^{-1} \sum_{i=1}^n X_jX_k /n$ as long as $w_{jk}\neq 0$. On the other hand, if $w_{jk}=0$, we set $\theta_{jk}=0$, since the corresponding score contribution for estimating $\theta_{jk}$ is exacly zero.

The TPL estimator $\hat\btheta$ of $\btheta$ is then obtained by applying the following thresholding mechanism to the entries of $\bS$:
\begin{equation}\label{eq:thetahat1}
	\hat \theta_{jk} =
	\left\{
	\begin{array}{ccc}
		S_{jk}, && \text{if }    \hat w_{jk} \neq 0, \\
		0, & &\text{if }    \hat w_{jk} = 0,
	\end{array}
		\right. \quad j,k =1,\dots, p\,,
\end{equation}
where $\hat \bw =(\hat w_{11}, \hat w_{12}, \dots, \hat w_{p-1p}, w_{pp})^\top$ is found by minimizing the empirical objective
\begin{align}\label{eq:crit_true2}
\hat d_\lambda(\bw) & =   \dfrac{1}{2} \bw^\top \hat \bJ \bw + \bw^\top \text{diag}\{ \hat \bJ\} +  \dfrac{\lambda}{n} \sum_{j<k} \frac{|w_{jk}|}{S_{jk}^2},
\end{align}
for some $\lambda \ge 0$. Here $\hat \bJ$ is the plug-in estimate of the $m\times m$ scores covariance matrix
\begin{align} \label{eq:hatJ}
\hat \bJ &= \dfrac{1}{n} \sum_{i=1}^n \bbU(\bs; \bX^{(i)})^\top  \bbU(\bs; \bX^{(i)}),
\end{align}
where $\bs = \text{vech}(\bS)$ and $\bU(\cdot; \bX)$ is the matrix containing all marginal and pairwise scores defined in (\ref{eq:score_pw}). In the rest of the paper, we use the notation $\hat \bTheta$ to denote the TPL estimator of $\bTheta$, such that $\hat \btheta = \text{vech}(\hat \bTheta)$.

Equation (\ref{eq:crit_true2}) represents a plug-in estimate of the population criterion $d_\lambda(\bw; \btheta)$  in Equation \eqref{eq:crit_true1_2}, obtained by setting $\btheta = \bs$ and with expectations replaced by averages. The threshold estimator in Equation (\ref{eq:thetahat1}) describes a simple mechanism whereby sparsity on $\btheta$ is induced by the relative information contained in the pairwise likelihood scores. Since $\bw$ is not, strictly speaking, a statistical parameter, but rather defines the structure of the estimating equations, selection of $\bw$ should require special considerations. Notably,  our objective is  different from that in classic regularization methods which obtain sparsity of $\btheta$ by maximizing the full or pairwise likelihood functions subject to a regularization constraint, such as the Lasso constraint $\sum_{j<k} |\theta_{jk}| < u$, $u>0$, on the parameter space.
 
In the remainder of this section, we derive an explicit expression for the minimizer of the empirical criterion $\hat d_\lambda(\bw)$ and provide sufficient conditions for its uniqueness. For a vector $\ba \in \mathbb{R}^m$, we use $\ba_{\mathcal{E}}$ to denote the sub-vector corresponding to index $\mathcal{E} \subseteq \{  jk : 1\leq j\leq k \leq p\}$, while $\bM_{\mathcal{E}}$
denotes the sub-matrix of the $m \times m$ squared matrix $\bM$ formed by taking rows and columns corresponding to $\mathcal{E}$.
 
\begin{theorem} \label{theory:uni}  The minimizer $\hat \bw$ of criterion $\hat d_\lambda(\bw)$ defined in (\ref{eq:crit_true2}) is unique with probability one, and given by
	\begin{equation}  \label{prop:uniqueness} 
	\hat \bw_{  \hatEps}  =  \hat \bJ_{\hatEps}  ^{-1} 
	\left[\mathrm{diag}\{\hat \bJ_{\hatEps}\} - \dfrac{\lambda}{n} \bEta \right],  \  \  \hat \bw_{   \setminus \hatEps}  = 0,
	\end{equation}
where $\bEta$ is the sub-gradient of the weighted $L_1$-norm $\sum_{j<k}   |w_{jk}|/S^2_{jk} $, that is 
\begin{align} \label{eq:eta_jk}
	\eta_{jk} =  \frac{1}{ S^{2}_{jk}}  \times 
	\left\{
	\begin{array}{ccl}
	    0, & &\text{if } j =  k, \\
		1, & &\text{if } j \neq k, w_{jk} >0, \\
		-1, & &\text{if }j \neq k,  w_{jk} <0, \\
		\in [ - 1, 1 ], & &\text{if } j \neq k,  w_{jk} = 0,
	\end{array}
	\right.
\end{align}
and 
\begin{equation}\label{eq:epsilonhat}
\hatEps = \left\{ jk   : \left\vert \mathrm{diag}\{\hat \bJ\}_{jk} - \hat \bJ_{jk, \cdot}\hat \bw    \right\vert \geq   \dfrac{\lambda}{n S^{2}_{jk}}\mathbb I(j \neq k)\right\},
\end{equation}
where $\setminus \hatEps$ denotes the complement index set $\{  jk : 1\leq j\leq k \leq p\} \setminus \hatEps$, $\mathbb I(\cdot)$ is the indicator function, $\hat \bJ_{jk, \cdot}$ is the row of $\hat \bJ$ corresponding to index $jk$ and $\mathrm{diag}\{\hat \bJ\}_{jk}$ is the element of $\mathrm{diag}\{\hat \bJ\}$ corresponding to index $jk$.
\end{theorem}

Positive definiteness of $\bTheta$ implies linear independence of elements in $\bX$, ensuring that the score elements defined in (\ref{eq:u_{jk}}) are nonzero with probability one. Thus by definition, $\hat \bJ_{\hatEps}$ is positive definite almost surely, leading to a unique solution of $\hat \bw$; in contrast, in Lasso regression the columns of the design matrix are required to be in general position to yield a similar result.

From Theorem \ref{theory:uni} it follows that, for sufficiently large $\lambda >0$, a relatively small subset of scores is selected. Formula (\ref{eq:epsilonhat}) suggests that the $jk$th score is selected when it contributes relatively large information  in the overall pairwise likelihood relative to the size of $S_{jk}$. The extent of the contribution is measured by the difference between the marginal estimated Fisher information  for the $jk$th score, $\mathrm{diag}(\hat \bJ)_{jk}$,  and the information already present in other selected scores, represented by the linear combination $\hat \bJ_{jk, \cdot}\hat \bw$.

\subsection{Relationship with adaptive thresholding} \label{sec:first_order}
Adaptive thresholding sets a covariance estimate to zero if the absolute value of the $z$-score $S_{jk}/ \text{SE}_{jk}$ is sufficiently small. In our setting, $\text{SE}_{jk} =  \{ n \widehat{\text{var}}(\bu_{jk}) \}^{-1/2}$,  where
$
\widehat{\text{var}}(\bu_{jk}) = n^{-1} \sum_{i=1}^n \vert\vert \bu _{jk}(\bs; \bX^{(i)}) \vert\vert^2_2
$ 
is the empirical Fisher information based on the the $jk$th pair.  While this approach is widely used for its extreme simplicity in large multivariate problems, $\text{SE}_{jk}$ clearly ignores the information already contributed by other covariance  estimates $S_{rs}$, $rs\neq jk$. An inspection of the first-order conditions associated with the objective (\ref{eq:crit_true2}) reveals that the sparse estimator $\hat \btheta$ is a type of adaptive thresholding estimator, with standard errors adjusted sequentially based on the information provided by other non-zero covariance estimates.

Inspecting the non-zero elements $\hat \bw_{  \hatEps}$ in the solution (\ref{prop:uniqueness}) provides further insights on the selection process. Consider first the case where $\hat w_{jk} \neq 0$; for this to be true we must have that
\begin{equation} \label{eq:KKTj}
	\sum_{i=1}^n \bu_{jk}(\bs; \bX^{(i)})^\top {\text{res}}^{(i)}_{jk} =      \lambda    \eta_{jk},        
\end{equation}
where ${\text{res}}^{(i)}_{jk}$ is the pseudo-residual defined by
\begin{equation*}
	{\text{res}}^{(i)}_{jk}  = \bu_{jk}(\bs; \bX^{(i)})  -  \sum_{st\neq jk} \bu_{st}(\bs; \bX^{(i)}) \hat w_{st}. 
\end{equation*}
Taking the absolute value on both sides of (\ref{eq:KKTj}) shows that a sufficient condition for $\hat w_{jk} \neq 0$ for some $j<k$ is
\begin{equation*} 
	S^2_{jk}  \times \left\vert   \sum_{i=1}^n \bu_{jk}(\bs; \bX^{(i)})^\top \text{res}^{(i)}_{jk} \right\vert \ge  \lambda \,,
\end{equation*}
and $\hat w_{jk} = 0$  otherwise. From the above expression, we can see that the final estimator $\hat \btheta$ follows the   adaptive thresholding mechanism
\begin{equation}\label{eq:thetahat}
	\hat \theta_{jk} =
	\left\{
	\begin{array}{ccc}
		S_{jk}& & \text{if }     |S_{jk}|/\text{SE}^{\text{adj}}_{jk} \ge \sqrt{\lambda}, \\
		0& & \text{if }    |S_{jk}|/\text{SE}^{\text{adj}}_{jk} < \sqrt{\lambda},
	\end{array}
	\right.
\end{equation}
where  $\text{SE}^{\text{adj}}_{jk}$ is the adjusted standard error
\begin{equation}\label{eq:adj_error}
	\text{SE}^{\text{adj}}_{jk} = n^{-1/2}\left\{ \widehat{\text{var}}(\bu_{jk}) - \widehat{\text{cov}}(\bu_{jk}, \hat \bw^\top \bu ) \right\}^{-1/2} \,,
\end{equation}
involving the quantities
\begin{align*}
	\widehat{\text{var}}(\bu_{jk})  &=  \dfrac{1}{n}\sum_{i=1}^n \left\Vert \bu_{jk}(\bs; \bX^{(i)}) \right\Vert^2_2,  \\  
	\widehat{\text{cov}}(\bu_{jk}, \hat \bw^\top \bu  ) &=   \dfrac{1}{n}  \sum_{i=1}^n \bu_{jk}(\bs; \bX^{(i)})^\top  \sum_{st \neq jk}  \hat \bw_{st} \bu_{st}(\bs; \bX^{(i)}) \,.
\end{align*}
with the last expression representing the covariance between the $jk$th score and the pairwise likelihood based on scores different from  $\bu_{jk}$.

The adjusted standard error expression in (\ref{eq:adj_error}) emphasizes an important difference between adaptive thresholding method and our approach.   Unlike the usual standard error formula $\text{SE}_{jk} =  \{ n \widehat{\text{var}}(\bu_{jk}) \}^{-1/2}$, the adjusted standard error $\text{SE}^{\text{adj}}_{jk}$  is computed by removing the portion of variance  already explained by the linear combination of currently selected scores.

\subsection{Selection of $\lambda$}\label{sec:lambdaSel}

The tuning parameter $\lambda$ plays an important role in determining the proportion of nonzero elements in $\hat \bTheta$. The inequality in (\ref{eq:thetahat}) resembles a rejection region of a hypothesis test and  suggests that $\lambda$ may be chosen based on some form of error control. In this section, we describe a data-driven criterion to select $\lambda$  based on a sequential testing  of the null hypotheses of the form  $H_0$: $\theta_{jk}=0$ against the alternative $H_1$: $\theta_{jk}\neq 0$, for all $jk \in \hatEps$.

Under the null hypothesis that $\theta_{jk} = 0$, $\sqrt{n} S_{jk}$ converges in distribution to the univariate normal distribution $N(\theta_{jk}, v_{jk})$ where   $v_{jk} = \theta_{jk}^2 + \theta_{jj}\theta_{kk}$, under typical regularity conditions. Thus, $n S^2_{jk}/v_{jk}$ converges in distribution to $\chi^2_1$, a chi-square random variable with one degree of freedom. Since $S_{jk} \overset{p}{\to} \theta_{jk}$, we have
\begin{align}\label{nulldis}
    \dfrac{n S^2_{jk}}{S_{jk}^2 + S_{jj}S_{kk}} \overset{d}{\rightarrow} \chi^2_1,
\end{align}
by Slutsky's Theorem. The alternative hypothesis $H_1: \theta_{jk} \neq 0$ is rejected if $n S^2_{jk}/(S_{jk}^2 + S_{jj}S_{kk}) > q_\alpha$, where $q_\alpha$ refers to the $(1-\alpha)$-quantile of the chi-square distribution with one degree of freedom and $0<\alpha<1$ is user-specified.  This leads to the following rule:
\begin{align} \label{eq:sel_lambda}
\hat \lambda  = \inf\left\{ \lambda:  \dfrac{n S^2_{jk}}{S_{jk}^2 + S_{jj}S_{kk}}> \gamma, \text{ for all } jk \in \hat \cE  \right\}, 
\end{align}
where $\gamma = q_\alpha$ for some user specified significance level $\alpha$.

By the nature of the constrained quadratic optimization problem defined in (\ref{eq:crit_true2}),  more $w_{jk}$ are allowed to be nonzero as $\lambda$ decreases. Criterion (\ref{eq:sel_lambda}) can be interpreted as choosing the smallest $\lambda$ to allow as many as possible nonzero elements in $\hat\btheta$, while ensuring that they all test significantly differently from 0. Differently from  classical $L_1$-penalized estimation where $\lambda$ is a tuning parameter, note that here  $\lambda$ is a random variable once $\gamma$ is determined.

\subsection{Cooordinate descent computing}

In the following, we outline  a coordinate descent algorithm to minimize the convex objective function in \eqref{eq:crit_true2} and obtain $\hat\bw$.  The algorithm cycles through the elements $\bw = (w_{11}, w_{12}, \dots, w_{p-1p}, w_{pp})$, updating one element at a time until convergence. Let $[t]$ denote the  $t$th iteration of the algorithm; given the current estimate of the coefficients $\hat\bw^{[t]} = (\hat w^{[t]}_{11}, \hat w^{[t]}_{12}, \dots, \hat w^{[t]}_{p-1p}, w^{[t]}_{pp})$, we compute
\begin{equation}\label{eq:coordDesc}
\hat{w}^{[t+1]}_{jk} = 
\left\{
\begin{array}{ccc}
 \frac{1}{\hat\bJ_{jk,jk}}\left(\hat\bJ_{jk,jk} - \sum_{rs \neq jk} \hat{w}^{[t]}_{rs} \hat\bJ_{jk,rs}\right), & &\text{if } j = k \\
 \frac{1}{\hat\bJ_{jk,jk}}\mathcal{S}\left(\hat\bJ_{jk,jk} - \sum_{rs \neq jk} \hat{w}^{[t]}_{rs} \hat\bJ_{jk,rs} ; \frac{\lambda}{S^2_{jk}} \right), & & \text{if } j \neq k
\end{array}
\right.
\end{equation}
while keeping all the other elements $\hat{w}^{(t)}_{rs}$, $rs \neq jk$ fixed, until a stopping criterion is met. Here $\mathcal{S}(x; \lambda)$ denotes the soft-thresholding operator  $\mathcal{S}(x; \lambda) = \text{sign}(x)(\vert x\vert - \lambda)_+$, and $x_+ = \text{max}\{0, x\}$, with $\text{sign}(x)$ being the sign function taking values $-1$, $0$ and $1$ if $x <0$, $x = 0$ and $x > 0$ respectively. Operationally, the selection criterion for $\lambda$ described by Equation (\ref{eq:sel_lambda}) is implemented through a golden search rule, a univariate numerical optimization method effective in finding a minimum for unimodal functions on a specific interval. Through iterative steps, the procedure progressively narrows down the search interval for $\lambda$ until convergence. Our numerical experiments have shown  that this approach is more precise and computationally efficient compared to costly top-down grid searches.

\section{Properties}\label{sec:prop}

In this section, we investigate the asymptotic properties of the selection procedure and the behavior of the implied estimator $\hat \btheta$ defined in (\ref{eq:thetahat1}). We simplify notations by letting $\bJ=\bJ(\btheta)$ be the covariance of scores at the true parameter value, $\bh = \text{diag}(\bJ)$ and $\hat \bh = \text{diag}(\hat \bJ)$. We use $\| A \|_1$, $\| A \|_\infty$ and $\| B \|_{max}$ to denote the $L_1$ norm, infinity norm and the max norm of vector $A$ and matrix $B$.
Define $\xi =\{ jk: \theta_{jk}\neq 0 \text{ and } j<k\}\cup\{jj: \theta_{jl}\neq 0 \text{ for some } l\neq j\}$ as the index set for the nonzero covariances and corresponding margins. Let $\hatxi =\{ jk: \hat\theta_{jk}\neq 0 \text{ and } j<k\}\cup\{jj: \hat\theta_{jl}\neq 0 \text{ for some } l\neq j\}$ be the estimate of $\xi$.

Let $m_0 = |\Eps\setminus \{jj, j=1\ldots,p\}|$ be the number of nonzero bivariate covariances in $\btheta$. We consider two main scenarios: the first where the data dimension $p$ is fixed, but no particular distributional assumption on $\bX$ is made; the second where $p$ is allowed to grow faster than the sample size $n$, but we require restrictions on the tail behavior of the data generating process. In the latter setting, we also allow $m_0$ to grow with $p$. For the results presented in this section, we require the following regularity conditions.
\begin{enumerate}[label=\textbf{C\arabic*},ref=C\arabic*]
\item[\bf C1]  The  elements of the covariance matrix $\bTheta$ are uniformly bounded. Moreover, the nonzero elements of $\bTheta$ are bounded away from 0.
\item[\bf C2]  The elements of $\bJ$, $\bJ_\xi^{-1}$ and $\bJ_\xi^{-1}\bh_\xi$ are uniformly bounded.
\item[\bf C3]   There exists some constant $c$ for all $m_0\geq 1$, such that $\min_{x}\|\bJ_{\xi,\xi_{-l}}x - \bh_\xi \|_1 > c$ for any $l=1,\ldots,\vert \xi \vert$, where x represents some length $\vert \xi \vert -1$ vector and $\bJ_{\xi,\xi_{-l}}$ is a $\vert \xi \vert\times (\vert \xi \vert-1)$ submatrix of $\bJ_\xi$ obtained by removing its $l$th column.
\item[\bf C4]   There exist some rate $r>0$ and some constant $c>0$, such that $\Pr(|X_j|>t)< c \cdot \exp(-t^r)$ for all $t\ge 0$ and $j=1,\ldots,p$.
\end{enumerate}

Condition C1 prescribes a regular behavior of the target covariance matrix  $\bTheta$ needed when $p$ diverges, while Condition C2 guarantees the existence of   $\bJ_\xi^{-1}\bh_\xi$,  i.e., the minimizer of the ideal population objective $d_\lambda(\bw; \btheta)$, at $\lambda = 0$, limited to those entries of $\bw$ corresponding to zero entries in $\bTheta$. Condition C3 is a main requirement in all our results and states that the columns of $\bJ_\xi$ are linearly independent and any $\vert \xi \vert-1$ columns of $\bJ_\xi$ are linearly independent of $\bh_\xi$, that is any linear combination of these columns cannot represent $\bh_\xi$ up to arbitrary precision. For the high-dimensional setting, we also use Condition C4 which assumes an exponential tail behavior for the entries of the random vector $\bX$. This includes variables with tails much heavier than sub-Gaussian tails.

In the next theorem, we begin to study the selection properties of our procedure in the simpler setting where $p$ and $n \rightarrow \infty$. In this case, only Conditions C1-C3 are required to ensure both selection consistency and entry-wise convergence in probability of $\hat \bTheta$ to $\bTheta$.
 
\begin{theorem}[Selection consistency for fixed $p$] \label{theory:fixedp}
    Under Conditions C1, C2 and C3, for fixed $p$ and $m_0$, there exists some constant $c$, such that when $\gamma=cn$, we have
    $
        \Pr(\hatEps = \Eps) \to 1
    $, as $n \to \infty$.
\end{theorem}

When $p$ and $m_0$ are fixed and $n$ grows, selection consistency occurs under rather mild conditions. First note that Theorem \ref{theory:fixedp} does not make restrictive assumptions on the distribution of the original random vector $\bX$ or on its moments. The main requirement here is Condition C3, which plays a role analogous to that of the Irrepresentable Condition (IC) commonly used in the context of Lasso regression and other similar sparsity-inducing penalization methods for regression; e.g., see Zhao and Yu  \cite{zhao2006model}.  Our version of the IC asserts that the scores corresponding to zero covariances  cannot be represented as linear combinations of the scores for the non-zero covariances. The condition is needed to ensure that our procedure correctly identifies the true set of non-zero covariances while not being misled by correlations between scores.

\begin{theorem}[Local false positive error rate for large $p$] \label{theory:false_positive}
    Under Conditions C1-C3, for arbitrary $p$, letting $\gamma^{-1} = o(n^{-1})$, we have   $
    \Pr(jk \in \hatEps \  | \ jk \in \setminus \Eps) \to 0$,  as $n \to \infty$.
\end{theorem}

Theorem \ref{theory:false_positive} shows that the probability of a  local false positive error -- or Type I error -- for an individual covariance entry becomes negligible as $n$ increases.  Moreover, the TPL  selection procedure controls the  false positive errors locally  in large-sample scenarios, irrespective of the data size $p$, as long as the threshold $\gamma$ diverges quickly enough with $n$.

Next, we turn to selection consistency when the data dimension $p$ grows exponentially with the sample size $n$.  The next theorem is organized into two main parts.  Part (i) states that the probability of the overall false positive rate converges to zero. This property implies the sparsistency property, meaning that the TPL estimator identifies consistently the zero covariance entries as the sample size increases. Part (ii) addresses the power of our selection method in identifying the true nonzero covariances. Specifically, the probability that all true nonzero covariances are correctly identified  approaches 1. Together, these results confirm the selection consistency of the estimator.

\begin{theorem}[Selection consistency for large $p$] \label{thm:consistency}
    Under Conditions C1-C4, for any constant $r_1 \in (0, \frac{r}{8+r})$ and $r_2= \min(\frac{r}{32+4r}, \frac{4+r}{32+4r}-\frac{r_1}{4})$, let $p=o\{\exp(n^{r_1-c})\}$ and $m_0=o(n^{r_2-c})$ for some arbitrarily small $c\in (0, \min\{r_1,r_2\} )$.   Then when $\gamma = O(n^{1-2r_2})$ and $\gamma^{-1} = O(n^{-r_1-\frac{4}{8+r}})$,  we have:
   \begin{itemize}
    \item[(i)] (Global false positive error rate) $\Pr(\hatEps \cap \setminus \Eps = \emptyset) \to 1$, as $n\to \infty$.
    \item[(ii)] (Power) $\Pr(\Eps \subseteq \hatEps ) \to 1$, as $n\to \infty$.
   \end{itemize}
 Hence, we have $\Pr(\hatEps = \Eps ) \to 1 $, as $n\to \infty$.
\end{theorem}

A few remarks on this result are in order.  In the increasing $p$ setting selection consistency hinges on Condition C4, which assumes exponential tail behavior for the entries of the random vector $\bX$ with rate $r$. This assumption is milder and more general compared to other common methods often requiring sub-Gaussian random variables (see, e.g., Bickel and Levina \cite{bickel2008covariance} and Lam \cite{lam2009sparsistency}). Model selection consistency for the TPL selection is thus understood to apply to a broader range of distributions encountered in practical applications, including Gamma-type distributions and those with Gaussian tails.

Differently from Theorem \ref{theory:fixedp} where $p$ and $m_0$ are fixed, the quantities $r_1>0$ and $r_2>0$ determine the maximum rates allowed for $p$ and $m_0$, respectively, with $r_1$ being negatively related to $r_2$. Lighter tails for the data distribution imply  that one can afford a larger data dimension $p$. Moreover, as expected, it should be easier to handle a larger proportion of non-zero elements when $p$ grows slowly, and vice versa. To capture this behavior, instead of fixing $r_1$ at its maximum, Theorem \ref{thm:consistency} allows one to specify the trade-off between sparsity and total data dimension by choosing $r_1$ within the range $0 < r_1 < \frac{r}{8+r}$, and allowing $r_2$ to increase accordingly. For example, if  $\bX$ is Gaussian, the tail of each component decays at rate $\exp(- t^2/2)$ as $t$ increases, meaning that $r=2$  and $0<r_1 < 1/5$. At the maximal rate allowed for $p$,  $r_1 = 1/5$ and $\gamma$ should be between $n^{0.6}$ and $n^{0.9}$. On the other hand, at the minimal rate when $r_1$ is arbitrarily close to zero,  $\gamma$ should be between $n^{0.4}$ and $n^{0.9}$.

The penalty coefficient $\lambda$ in our problem is a random variable depending on the choice of the tuning parameter $\gamma$, as defined in \eqref{eq:sel_lambda}. This is an important  difference compared to other  $L_1$-penalization methods, where the tuning constants controlling the sparsity level are fixed or treated as sequences typically vanishing as $n$ grows. Simple computations show that $1-2r_2 > r_1 + \frac{4}{8+r}$, ensuring that $\gamma$ always exists. This means that one can always choose $\gamma$ appropriately so that both the non-zero and zero elements in the covariance matrix can be correctly identified according to Parts (i) and (ii) of Theorem \ref{thm:consistency}.

We conclude this section by studying the behavior of the estimator implied by our selection procedure. Denote $\hat \bTheta^{oml}$ the oracle ML estimator of the covariance matrix which assumes the knowledge of no-zero entries of $\bTheta$, with elements $\hat \bTheta^{oml}_{jk} = S_{jk}$ if $\theta_{jk}\neq 0$, and $\hat \bTheta^{oml}_{jk} = 0$ if $\theta_{jk}=0$. Next, we compare the TPL estimator to $\hat \bTheta^{oml}$.
\begin{corollary}[Convergence to the oracle ML estimator] \label{corollary:mle} Under the assumptions of Theorem \ref{thm:consistency}, we have $\Pr(  \hat \bTheta = \hat \bTheta^{oml}) \to 1$, as  $n \to \infty$.
\end{corollary}

Corollary \ref{corollary:mle} states that, as the sample size grows, our estimator converges in probability to the oracle ML estimator, as if the exact locations of the nonzero correlations were known in advance, despite the large size of the original covariance matrix and the initially unknown locations of these nonzero correlations. Furthermore, it can be easily shown that $r_2$ is always positive, implying that $m_0$ is allowed to grow geometrically with $n$ to achieve the optimal rate in the worst-case scenario. In comparison, Lam and Fan \cite{lam2009sparsistency} demonstrated that, due to the control of bias, $L_1$-penalized likelihood approaches are limited to an $O(1)$ number of non-zero covariance elements in the worst-case scenario to obtain the optimal rate, regardless of how fast or slow $p$ grows.  Nonetheless, since we allow $p$ to grow exponentially with $n$, $m_0$ cannot grow at the rate $O(p)$ stated in Lam and Fan \cite{lam2009sparsistency}. This constraint depends on our choice of the score covariance estimator $\hat \bJ$ defined in (\ref{eq:hatJ}), which would become nearly singular under selection consistency, because $m_0$ would be much greater than the sample size $n$.

\section{Numerical examples}\label{sec:num}
\subsection{Monte Carlo experiments} \label{sec:MC}

In this section, we study the finite-sample properties of the TPL estimator focusing on its support-recovery performance. We draw $n$ i.i.d. samples from the $p$-variate normal distribution $N_p(\bf0, \bTheta)$ under different  structures for $\bTheta$ by varying the sparsity level  $\tau$, defined as the proportion of zero  off-diagonal elements in $\bTheta$. The considered covariance structures are described as follows, and are also depicted in Figure \ref{fig:covStruct}. (i) {\it Block diagonal structure}: $\bTheta_{jj} = 1$, for $j = 1,\dots,p$, while off-diagonal elements are zero except for a block of adjacent variables, and whose dimension depends on $\tau$. The entries of the off-diagonal elements of the nonzero block are drawn independently from the univariate normal distribution $N(0.5, 0.05^2)$. Additionally, we also consider first-order autoregressive model to determine the elements in the non-null block; the results are similar to those for the block diagional structure and are therefore omitted here. (ii) {\it Sparse at random structure}. The non-zero entries of $\bTheta$ are generated at random from an Erd\H{o}s-Rényi (ER) graph model with probability of connection $\tau$. The covariance matrix $\bTheta$ is obtained by finding a positive definite matrix closest to the sample covariance on the given ER random graph by the iterative conditional fitting algorithm of Chaduri et al. \cite{chaudhuri2007estimation}. This structure is challenging for support recovery, since $\bTheta$ tends to have many nearly zero entries among the non-zero entries.

\begin{figure}[t]
  \hspace*{\fill}%
  \subcaptionbox{}{\includegraphics[scale = 0.48]{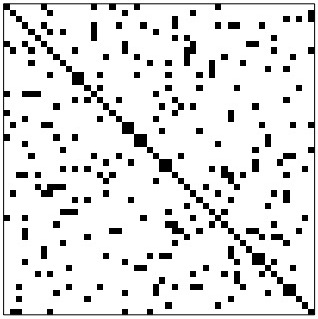}}\hspace{3em}%
  \subcaptionbox{}{\includegraphics[scale = 0.48]{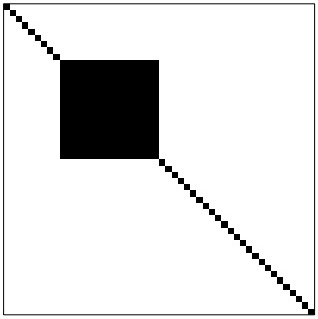}}%
  \hspace*{\fill}%
  \caption{Examples of sparse at random (left) and block diagonal (right) structures for $\bTheta$ for $p = 50$ and sparsity level $\tau = 0.9$. Black squares represent non-zero entries.}
  \label{fig:covStruct}
\end{figure}

We consider  $n = 40,100,250$, $p = 20,50,150$ and $\tau = 0.5,0.9$. which allow us to explore both $n>p$ and $n<p$ scenarios. The performance of the method explored in both sparse and relatively dense scenarios, with the latter deviating somewhat from the typical high-dimensional setting. For each combination of $n$, $p$  and $\tau$, we generate $B = 100$  Monte Carlo samples. The tuning constant $\lambda$ is selected according to (\ref{eq:sel_lambda}) with $\alpha = 0.1$ in all our simulations. The reported analyses have been conducted in the \texttt{R} environment \citep{RCoreTeam}, with some subroutines  written in \texttt{C++} to reduce the computing time. The support recovery performance of the TPL method is assessed using the following sensitivity (SN), specificity (SP) and  accuracy (AC)  metrics:
\begin{eqnarray*}
\text{SN} &=& \frac{\sum_{j<k} I( \hat\theta_{jk} \neq 0, \theta_{jk} \neq 0 )}{m_0}, \;\;\;
\text{SP} = \frac{\sum_{j<k} I(\hat\theta_{jk} = 0, \theta_{jk} = 0)}{ m-p-m_0 }, \\
\text{AC} &=& \frac{\sum_{j<k}I(\hat\theta_{jk} \neq 0, \theta_{jk} \neq 0) + \sum_{j<k} I( \hat\theta_{jk} = 0, \theta_{jk} = 0)}{p(p-1)/2}.
\end{eqnarray*}

\begin{table}[t]
\centering
\caption{Monte Carlo estimates of sensitivity (SN), specificity (SP) and accuracy (AC) for the truncated pairwise likelihood (TPL) estimator. Results are based on $B = 100$ samples of size $n$ from a $p$-variate normal distribution $N_p(\mathbf{0}, \bTheta)$, where $\bTheta$ has low ($\tau = 0.5$) or high sparsity ($\tau = 0.9$) according to  block diagonal or random structures.}
 \addtolength{\tabcolsep}{-3pt}
\begin{tabular}{cccccccccccccc}
\hline
\multicolumn{1}{c}{} & & \multicolumn{1}{c}{} & \multicolumn{3}{c}{{SN}} && \multicolumn{3}{c}{{SP}} && \multicolumn{3}{c}{{AC}} \\
\cline{4-14}
$\tau$ & $p$ & & \multicolumn{3}{c}{$n$} && \multicolumn{3}{c}{$n$}  && \multicolumn{3}{c}{$n$} \\
& & & 40 & 100 & 250 && 40 & 100 & 250 && 40 & 100 & 250\\
\cline{4-14}
     & &  & \multicolumn{11}{c}{{ Block diagonal}} \\
     & 20  && 0.88 & 0.99 & 1.00 && 0.97 & 0.95 & 0.92 && 0.93 & 0.97 & 0.96\\
0.5  & 50  && 0.85 & 0.99 & 1.00 && 0.98 & 0.96 & 0.94 && 0.92 & 0.98 & 0.97\\
     & 150 && 0.74 & 0.99 & 1.00 && 0.99 & 0.97 & 0.95 && 0.86 & 0.98 & 0.97\\
\\
     & 20  && 0.89 & 0.99 & 1.00 && 0.96 & 0.94 & 0.92 && 0.96 & 0.95 & 0.93\\
0.9  & 50  && 0.82 & 0.99 & 1.00 && 0.98 & 0.96 & 0.94 && 0.96 & 0.96 & 0.94\\
     & 150 && 0.73 & 0.99 & 1.00 && 0.99 & 0.97 & 0.95 && 0.96 & 0.97 & 0.95\\
     & & & \multicolumn{11}{c}{{ Sparse at random}} \\
     & 20  && 0.36 & 0.46 & 0.43 && 0.97 & 0.96 & 0.95 && 0.67 & 0.70 & 0.70\\
0.5  & 50  && 0.23 & 0.24 & 0.29 && 0.98 & 0.97 & 0.95 && 0.60 & 0.62 & 0.62\\
     & 150 && 0.11 & 0.14 & 0.15 && 0.99 & 0.97 & 0.95 && 0.55 & 0.55 & 0.65\\
\\
     & 20  && 0.66 & 0.67 & 0.55 && 0.96 & 0.95 & 0.93 && 0.93 & 0.92 & 0.90\\
0.9  & 50  && 0.34 & 0.31 & 0.36 && 0.98 & 0.96 & 0.94 && 0.91 & 0.89 & 0.88\\
     & 150 && 0.15 & 0.18 & 0.22 && 0.99 & 0.96 & 0.95 && 0.90 & 0.89 & 0.88\\
\hline
\end{tabular}
\label{tab:suppRec}
\end{table}

Table \ref{tab:suppRec} shows Monte Carlo estimates of SN, SP and AC statistics  for both block diagonal and sparse at random structures. For the block diagonal structure, the support recovery looks generally favorable: all the metrics appear to be rather insensitive to changes in $n, p$ and $\tau$, attaining excellent results across all the settings. Particularly, the overall accuracy remains stable regardless of $\tau$ and $p$. While the specificity generally improves with $p$, a slight degradation in sensitivity is observed when $n = 40$ and $p = 150$. Interestingly, in this setting the TPL estimator maintains desirable properties even with low sparsity ($\tau=0.5$), regardless of the data dimension, pointing out the applicability of the methodology in scenarios beyond sparse estimation.

Interesting observations can be drawn from the more challenging sparse at random setting, where many elements nearly zero make it harder to reconstruct the exact covariance structure. Although specificity generally increases with $p$, sensitivity is not entirely satisfactory, especially with $\tau = 0.5$ and $p$  large. While possibly disappointing, this behavior is linked to the relative efficiency of the estimator. Particularly,  scenarios where the sensitivity is lower correspond also to ones where $\hat{\bTheta}$ outperforms the oracle ML estimator $\hat \bTheta^{oml}$, our method tends to over-sparsify, especially in dense scenarios, by shrinking small but non-zero elements to zero. This results in many false negatives, reducing sensitivity. However, over-sparsification significantly reduces variance, thus reducing the overall mean squared error.

\subsection{Mouse brain single-cell RNA data} \label{sec:real_data}

Single-cell transcriptome sequencing (scRNA-seq) is a powerful technique that enables RNA sequencing from individual cells and allows the study of their specific gene expression profiles \citep{kolodziejczyk2015technology}. We illustrate the TPL estimtion procedure by  analyzing a dataset collected by Zeisel et al. \cite{zeisel2015cell} to provide a census of two brain regions, namely the hippocampal CA1 and the primary somatosensory cortex. The dataset is available in the Gene Expression Omnibus database (accession number \texttt{GSE60361}) and consists of $n = 2816$ mouse brain cells for which the expression of $p = 19839$ genes has been measured. In our analysis, we focus on two types of immune cells,  \emph{microglia} and \emph{perivascular macrophages}, corresponding to $n_\texttt{mgl} = 28$ and $n_\texttt{pvm} = 50$ observations, respectively. Although they originate from distinct developmental pathways, these two cell types are closely related and play a crucial role in brain development, homeostatsis regulation and  neurodegenerative diseases. In this context,  sparse covariance estimates can disentangle complex gene interactions and offer visual summaries through their link to covariance graph models and gene co-expression networks \citep{chaudhuri2007estimation}. Data quality checks are conducted following Amezquita et al. \cite{amezquita2020orchestrating}, while data normalization, pre-processing and feature filtering steps follow Townes et al. \cite{townes2019feature}. The number of considered genes is therefore reduced to $p=200$. According to the literature, this is a conservative overestimate, as the number of informative genes in scRNA-seq is typically relatively small within a single tissue, due to limited biological variation.

In Figure \ref{fig:EstCovRes01}, we show covariance estimates for the two types of cells obtained by the TPL estimator. Consistently with our experiments on simulated data, we use $\alpha = 0.1$ to select $\lambda$ according to the criterion in \eqref{eq:sel_lambda}.
A first visual inspection highlights how the estimated sparsity structures are different. Although we considered the same $\alpha$, the sparsity pattern looks strongly dependent on the cell type. Particularly, the number of non-zero off-diagonal elements corresponds to 7.89\% and 13.93\% for  \emph{microglia} and \emph{perivascular macrophages} types, respectively. The non-zero correlations are generally positive, while  negative correlations are, in most of the cases, shrunk to zero. The estimated matrices differ not only in terms of sparsity magnitude, but also in terms of the location of the non-zero elements.  In Figure \ref{fig:EstCovResLambdaVar}, we show estimated covariance matrices for  \emph{microglia} for $\alpha = 0.01$ and $\alpha = 0.4$, which correspond to 2.58\% and 26.98\% non-zero elements in $\hat \bTheta$, respectively. This illustrates how the proposed thresholding approach can be used to explore different requirements in terms of sparsity magnitude by tuning $\alpha$. In turn, this is a starting point for further explorations which can provide some insight on the relationships among the genes in the considered cells.

\begin{figure}[t]
  \hspace*{\fill}%
  \subcaptionbox{ }{\includegraphics[scale = 0.31]{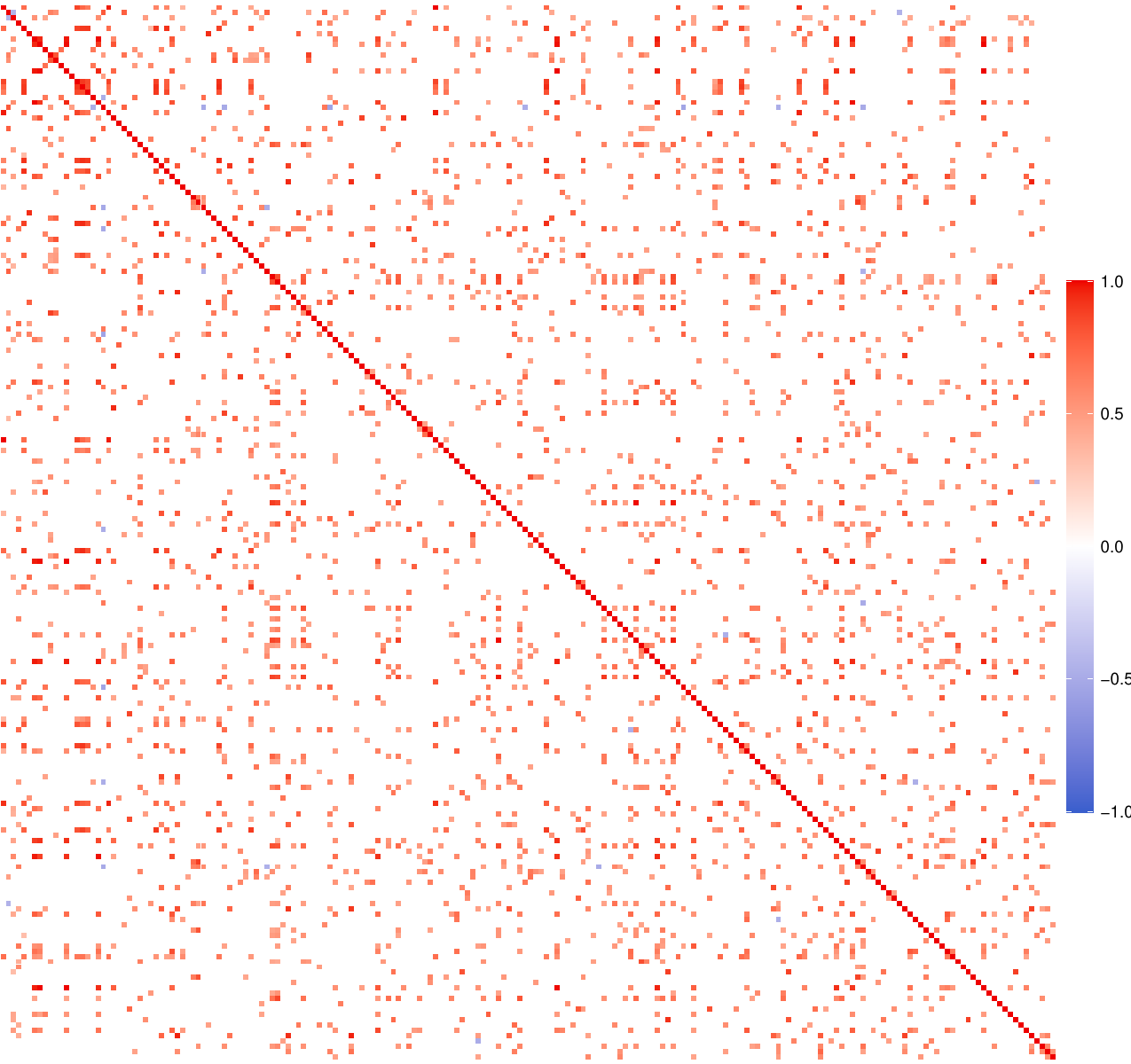}}\hspace{1em}%
  \subcaptionbox{ }{\includegraphics[scale = 0.31]{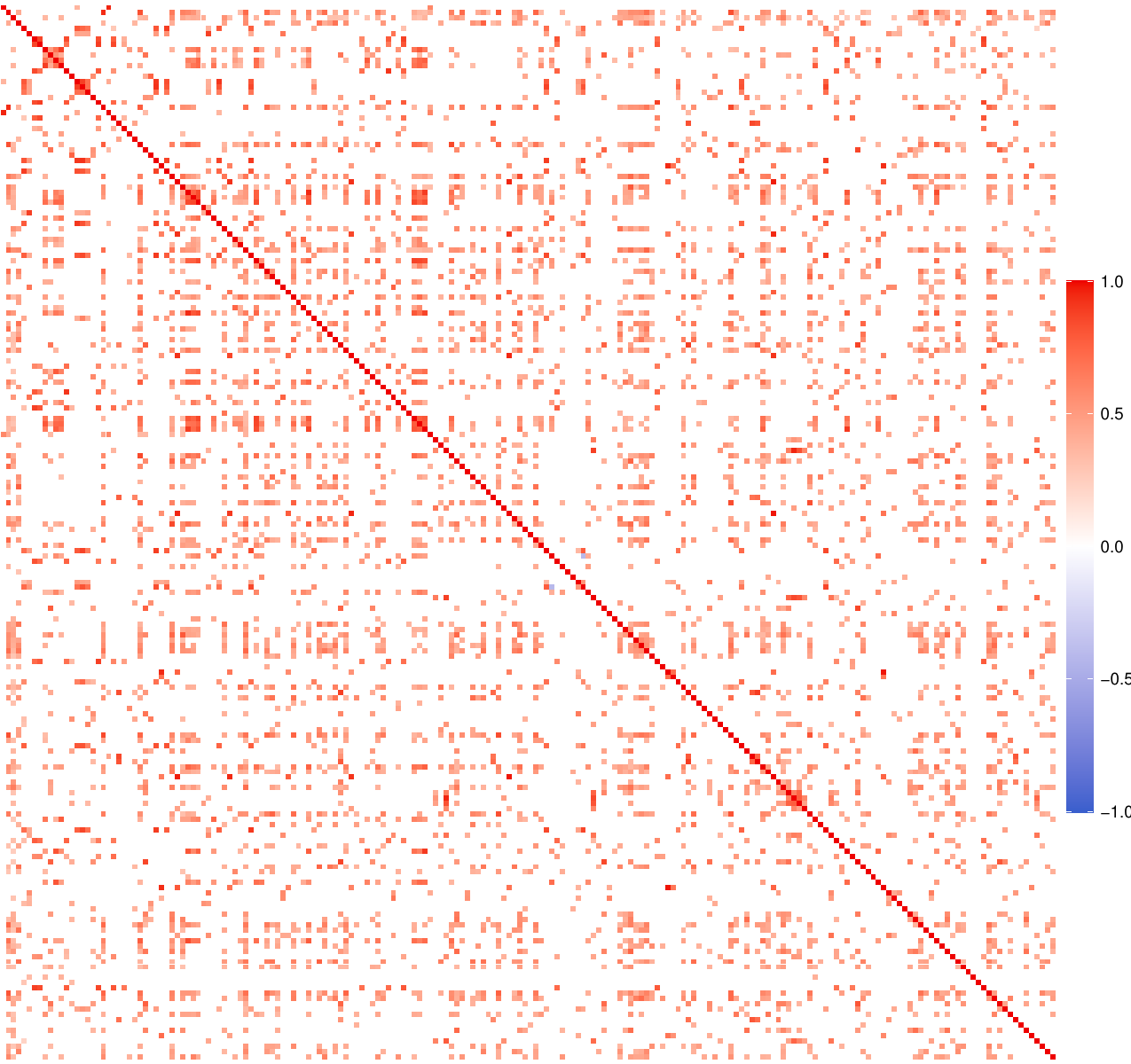}}%
  \hspace*{\fill}%
  \caption{Estimated sparse correlation matrices for \emph{microglia} cells (a) and \emph{perivascular macrophages} cells (b) using the TPL estimator with $\alpha = 0.1$. }
  \label{fig:EstCovRes01}
\end{figure}

\begin{figure}[t]
  \hspace*{\fill}%
  \subcaptionbox{ }{\includegraphics[scale = 0.31]{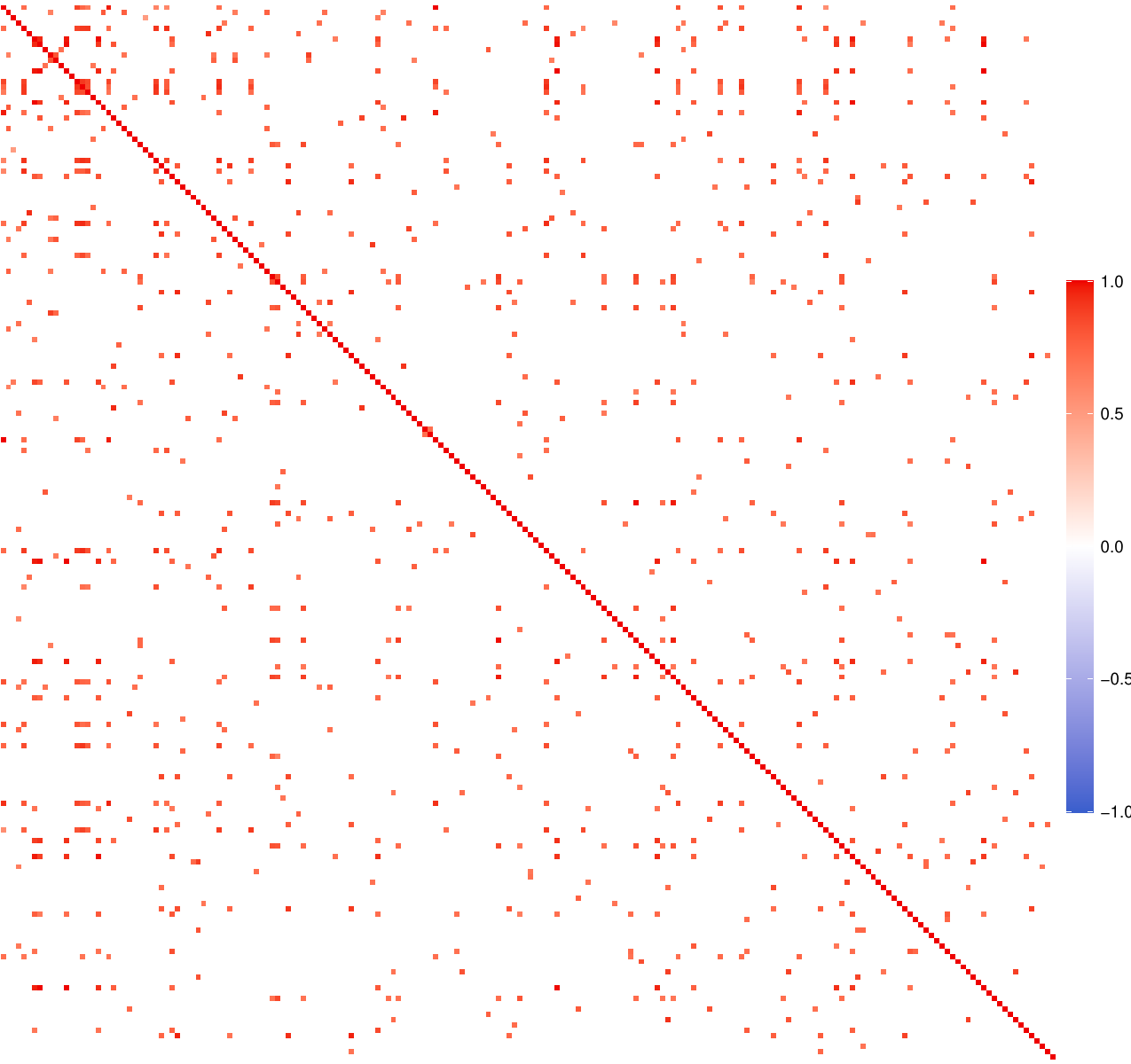}}\hspace{1em}%
  \subcaptionbox{ }{\includegraphics[scale = 0.31]{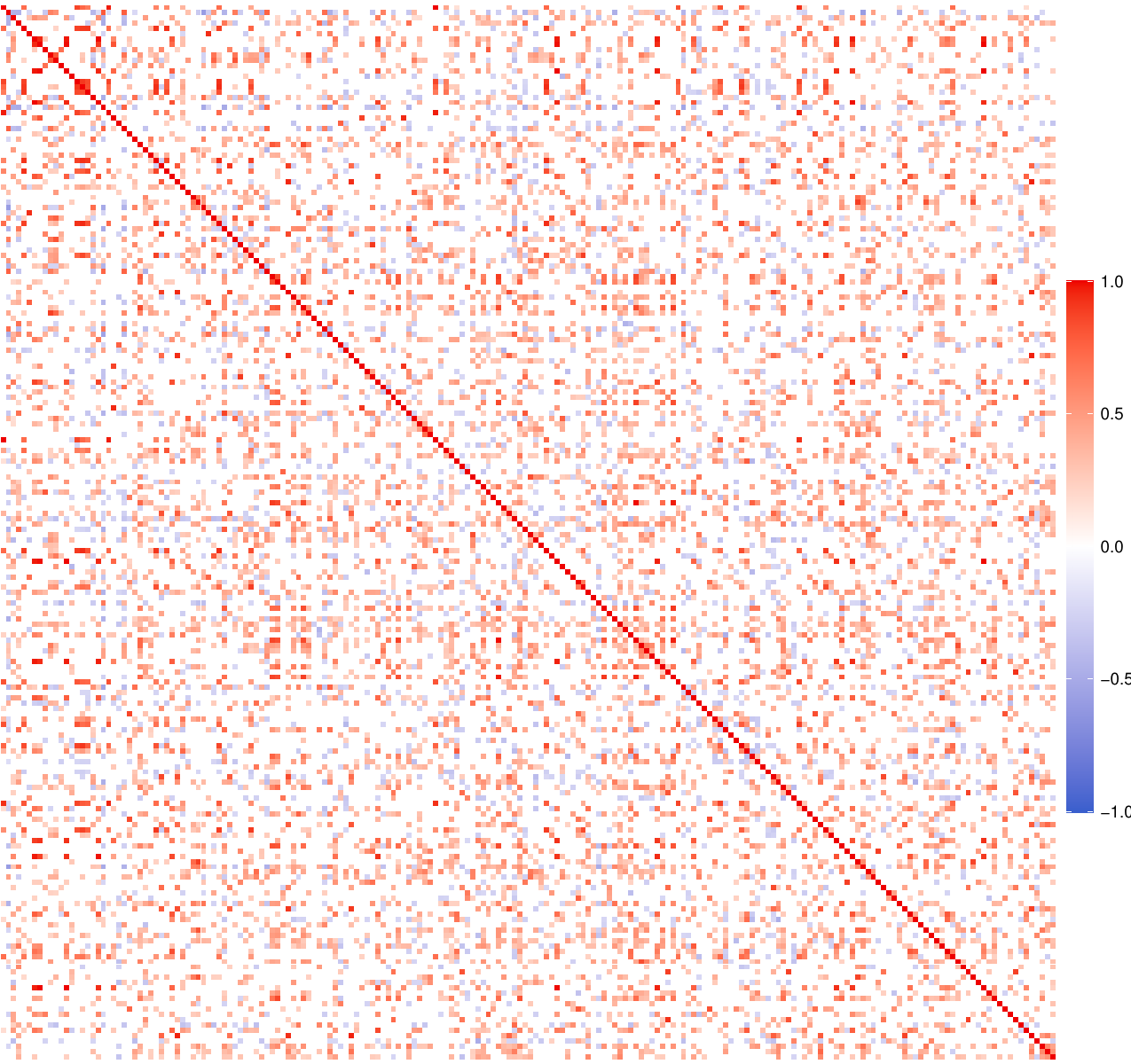}}%
  \hspace*{\fill}%
  \caption{Estimated sparse correlation matrices for \emph{microglia} cells using the TPL estimator $\alpha = 0.01$ (left) and $\alpha = 0.4$ (right). }
  \label{fig:EstCovResLambdaVar}
\end{figure}

\section{Conclusion and final remarks} \label{sec:conclusion}

Despite recent methodological advancements in high-dimensional covariance estimation, significant challenges remain in constructing unbiased estimators achieving optimal  statistical accuracy, while ensuring model selection consistency when $ p $ is much larger than $ n $. The primary contribution of this work is the development of a framework for high-dimensional covariance selection and estimation, featuring some distinctive traits of penalized likelihood and thresholding methodologies. The selection rule introduced in Equations (\ref{eq:crit_true})--(\ref{eq:crit_true1_3}) penalizes the discrepancy between the PL and ML score functions, thus aiming to maximize the statistical accuracy encoded by the PL score function for a given level of sparsity. Similar to Lasso-type approaches in regression, our method leverages the geometric properties of the $L_1$ penalty to set certain terms of the PL function to zero (see Theorem \ref{theory:uni}). The new criterion also functions as a form of adaptive thresholding, adjusting the marginal information that determines selection uncertainty based on the chosen covariance estimates (see Section \ref{sec:first_order}).

A key distinction of our TPL procedure from existing penalization approaches for high-dimensional estimation is its focus on selecting terms in the pairwise likelihood function, with each term corresponding to an entire bivariate model, rather than penalizing individual elements of $\bTheta$ in a likelihood criterion.
While penalized likelihood estimators rely on biased estimating equations to induce selection \cite{fan2001variable, zou2006adaptive}, an  advantage of our construction is that the resulting pairwise estimating equations are inherently unbiased. Therefore, conditional on consistent selection of the relevant variable pairs provided in Theorem \ref{thm:consistency}, the resulting covariance estimates remain unbiased and the  estimator $\hat \bTheta$ behaves asymptotically as the oracle ML estimator (Corollary \ref{corollary:mle}). In terms of model selection performance, we showed that the TLP procedure has the selection consistency property when $p$ grows exponentially in $n$ and the data generating process has arbitrary exponential tails. The selection consistency result stated in Theorem \ref{thm:consistency} is generally stronger than sparsistency alone - i.e., the zero covariances are correctly estimated with probability tending to one -- with the latter implied by Part (i) of  Theorem \ref{thm:consistency}. Since the conditions outlined in Theorem \ref{thm:consistency} also ensure power converging to one, they are more stringent than necessary to prove sparsistency alone.

The form of the solution given in Theorem \ref{theory:uni} and the discussion in Sections \ref{sec:first_order} reveals a connection between PL estimation and adaptive thresholding frameworks that, to our knowledge, has not been previously explored.  Specifically, the criterion function in Equation (\ref{eq:crit_true}) introduces a general method for constructing adaptive thresholding procedures, depending on the distributional assumption for  the pairwise score and the type of the penalty. For instance, when the variance elements on the diagonal of $\bTheta$ are fixed, the score covariance matrix $\hat \bJ$ becomes diagonal, meaning that the pairwise scores contain information on separate covariance parameters. Consequently, the resulting thresholding rule matches the adaptive thresholding method of Cai and Liu  \cite{cai2011adaptive}. In future research, sparse covariance estimation in the context of heavy-tailed data distributions may be achieved using pairwise scores based on bivariate t-distributions. Similarly, estimating the dependence structure for high-dimensional extreme events could be explored using bivariate extreme value distributions.

The numerical results presented in Section \ref{sec:num} corroborate the support recovery properties of the TPL procedure for sparse matrices in finite samples for different covariance structures. Additionally, they provide insight into the behavior of the estimator in the challenging case of sparse at random unstructured covariance with relatively low sparsity. The presence of many nearly zero covariance elements leads to general overshrinkage based on our naive choice of $\lambda$ based on the chi-square 0.9-quantile. While reducing the sensitivity of the method, this also helps avoid the accumulation of estimation errors, resulting in mean squared error considerably smaller than that of the oracle ML estimator. Overall this behavior suggests that further improvements are  possible through adaptive selection of $\lambda$. Particulary, the selection criterion in Equation (\ref{eq:sel_lambda}) suggests treating the selection of $\lambda$ as a top-down sequential hypothesis testing procedure for nested models. Following the approach of Lockhart et al. \cite{lockhart2014significance} in the context of sparse linear regression, one can order hypotheses according to a set of knots for $\lambda$, such as $\lambda_{max} > \dots > \lambda_{min} \ge 0$, each corresponding to the entry of a new pairwise score.  Starting from $\lambda_{max}$ corresponding to a diagonal $\hat{\mathbf{\Theta}}$, pairwise scores are added at each knot until a relatively dense estimate $\hat{\mathbf{\Theta}}$ is reached at $\lambda_{min}$. This opens avenues for future research on the interplay between hyperparameter selection and sequential model selection procedures, particularly in controlling false discovery rates in ordered testing settings  \cite{g2016sequential}.



\begin{appendix}

\section*{Appendix - Proofs of theorems}

\begin{proof}[\bf Proof of Theorem \ref{theory:uni}]
Note that $\bu_{jk}(\btheta; \bX)$ is the only sub-likelihood score with the element in position $\{(2p+2-j)(j-1)/2+k+1-j\}$  containing $\theta_{jk}$ for $j<k$. Since $\bTheta$ is positive definite, we have $\theta_{jk}^2 < \theta_{jj}\theta_{kk}$. This implies that $\bu_{jk}$ is a nonzero vector with probability one. By Formula (\ref{eq:score_mar}) and (\ref{eq:score_mar2}), all bivariate scores $\bu_{jk}$ are linearly independent. Thus, by the rank equality of Gram matrices, $\hat \bJ$ and $\hat \bJ_{\hat \varepsilon} $ are both full ranked.   
This implies that $\hat \bw$ exists and is unique due to  the strict convexity of the first term of Criterion $\hat d_\lambda(\bw)$ (\ref{eq:crit_true2}) and the convexity of the  second (penalty) term.
Let $\bu_{jk}^{(i)} = \bu_{jk}(\bs; \bX^{(i)})$ and $\bu^{(i)} = \bu(\bs, \hat \bw;\bX^{(i)})$. By the Karush-Kuhn-Tucker (KKT) conditions for quadratic optimization, the solution must satisfy
\begin{equation} \label{eq:kkt}
\sum_{i=1}^n {\bu_{jk}^{(i)}}^\top \bu^{(i)} - \sum_{i=1}^n {\bu_{jk}^{(i)}}^\top \bu_{jk}^{(i)} + \lambda S_{jk}^{-2}  \eta_{jk} =0,
\end{equation}
for $1\leq j \le k \leq p$, where $\eta_{jk}$ is defined in (\ref{eq:eta_jk}). The statement in the proposition then follows by solving the KKT equations in (\ref{eq:kkt}).
\end{proof}

\begin{lemma}\label{lem:hoeffding}
    Under Conditions C4, for any constant $r_1 \in (0, \frac{r}{8+r})$ let $p=o\{\exp(n^{r_1-c})\}$ for some arbitrarily small $c\in(0,r_1)$, then 
    \begin{itemize}
        \item[(i)] $\Pr( \max_i\| \bX^{(i)} \|^4_\infty > n^{\frac{4}{8+r}} ) \to 0$, as $n\to\infty$. 
        \item[(ii)] $\Pr( \max_i\| \bX^{(i)} \|^2_\infty > n^{\frac{2}{8+r}} ) \to 0$, as $n\to\infty$. 
    \end{itemize}
\end{lemma}
\begin{proof}[\bf Proof]
   For (i), Under Condition C4,
   \begin{align*}
       \Pr( \max_i\| \bX^{(i)} \|^4_\infty > n^{\frac{4}{8+r}} ) & \leq np \max_j \Pr(|X_j| > n^{\frac{1}{8+r}}) \\
       & = O_p\left\{ np \cdot \exp(-n^{\frac{r}{8+r}} ) \right\} \\
       & =o_p\left\{n \exp(n^{r_1-c}) \exp(-n^{\frac{r}{8+r}}) \right\}\\
       & =o_p \left[  n   \exp\left\{-n^{r_1}  (1 - n^{-c}) \right\}  \right ]\\
       & = o_p(1).
   \end{align*}
   Then $(ii)$ is implied by $(i)$.
\end{proof}

\begin{lemma}\label{lem:sconverge}
    Under Conditions C1-C4, for any constant $r_1 \in (0, \frac{r}{8+r})$ and $r_2= \min(\frac{r}{32+4r}, \\ \frac{4+r}{32+4r}-\frac{r_1}{4})$, let $p=o\{\exp(n^{r_1-c})\}$ and $m_0=o(n^{r_2-c})$ for some arbitrarily small $c\in (0, \min\{r_1,r_2\} )$, then $\| \bS - \bTheta \|_{max} = o_p(m_0^{-2})$.
\end{lemma}
\begin{proof}[\bf Proof]
    By Lemma \ref{lem:hoeffding} and the definition of $\bS$, the cores of $\bS$, $X^{(i)}_jX^{(i)}_k$ ($1\le j\le k\le p$), are uniformly bounded by $n^{\frac{2}{8+r}}$ with probability tending to one. Thus it suffices to prove the lemma when $\max_{i,j,k}\{X^{(i)}_jX^{(i)}_k\}<n^{\frac{2}{8+r}}$. We show the rest of the proof under this condition while omitting the conditioning notation for simplicity. 
    
    By the definition of $\bS$, since $\max_{i,j,k}\{X^{(i)}_jX^{(i)}_k\}<n^{\frac{2}{8+r}}$, we apply the Hoeffding's inequality for the centred U-statistic $S_{jk} - \bTheta_{jk}$ to have
    $$
    \Pr\left(| S_{jk} - \bTheta_{jk}|> n^{-2r_2} \right) \leq 2\exp\left \{ - 2n^{\frac{4+r}{8+r}} n^{-4r_2}  \right \} = 2\exp\left \{ -n^{r_1}  \right \}.
    $$
    Thus,
        $$
    \Pr(\max_{jk}| S_{jk} - \bTheta_{jk}|> n^{-2r_2} ) \leq 2p^2\exp\left \{ -n^{r_1} \right \} = o(1),
    $$
    since $p = o\{\exp(n^{r_1-c})\}$ for some $c>0$.
\end{proof}

\begin{lemma}\label{lem:jconverge}
    Under Conditions C1-C4, for any constant $r_1 \in (0, \frac{r}{8+r})$ let $p=o\{\exp(n^{r_1-c})\}$ for some arbitrarily small $c\in(0,r_1)$, then $\| \hat \bJ - \bJ \|_{max} = o_p(1)$.
\end{lemma}
\begin{proof}
    By Lemma \ref{lem:hoeffding}, $X^{(i)}_jX^{(i)}_kX^{(i)}_lX^{(i)}_s$ ($1\le j,k,l,s\le p$), are uniformly bounded by $n^{\frac{4}{8+r}}$ with probability tending to one. Thus it suffices to prove the lemma when $\max_{i,j,k,l,s}\{X^{(i)}_jX^{(i)}_kX^{(i)}_lX^{(i)}_s\} < n^{\frac{4}{8+r}}$. We show the rest of the proof under this condition while omitting the conditioning notation for simplicity. 

    Recall that 
    \begin{align} 
    \hat \bJ &= \dfrac{1}{n} \sum_{i=1}^n \bbU(\bs; \bX^{(i)})^\top  \bbU(\bs; \bX^{(i)}),
    \end{align}
    where $\bs$ is the vector containing sample covariance. To apply the Hoeffding's inequality for centred U-statistic, we let 
        \begin{align} 
    \tilde \bJ &= \dfrac{1}{n} \sum_{i=1}^n \bbU(\btheta; \bX^{(i)})^\top  \bbU(\btheta; \bX^{(i)}),
    \end{align}
    so that $\tilde \bJ$ is constructed by independent variables. By Lemma \ref{lem:sconverge}, elements in $\bS$ are uniformly bounded in probability, thus, it suffices to prove the lemma when it is the case. By the definition of $\bbU(\bs; \bX^{(i)})$, Condition C1 and Lemma \ref{lem:sconverge},
    \begin{align*}
        \| \hat J -\tilde J \|_{max} \le  c \cdot \max_{j,k,l,s}\left |  \frac{1}{n}\sum_{i=1}^n  X^{(i)}_jX^{(i)}_kX^{(i)}_lX^{(i)}_s   \right | \|\bS-\bTheta\|_{max},  
    \end{align*}
    for some constant $c>0$. Under Condition C4, the Hoeffding's inequality implies 
    $$
        \max_{j,k,l,s}\left |  \frac{1}{n}\sum_{i=1}^n  X^{(i)}_jX^{(i)}_kX^{(i)}_lX^{(i)}_s   \right | = O_p(1).
    $$
    Thus by Lemma \ref{lem:sconverge}, $\|  \hat J -\tilde J \|_{max} =o_p(m_0^{-2}) = o_p(1) $. 
    
    Although each $\bu_{jk}(\btheta,X)$ ($j<k$) is a length $m$ vector, it has at most three nonzero elements given in (\ref{eq:u_{jk}}).  Thus by definition, all elements of $\tilde \bJ$ are bounded by some constant times $n^{\frac{4}{8+r}}$ as stated at the beginning of the proof. The Hoeffding's inequality for the centred U-statistic implies that for any $\delta>0$,
        $$
    \Pr( \| \tilde \bJ - \bJ \|_{max}>\delta) \leq 2p^4\exp\left \{ - c \cdot n^{\frac{r}{8+r}} \delta^2  \right \} = o(1),
    $$
    for some $c>0$. Thus $\| \tilde \bJ - \bJ \|_{max} = o_p(1)$, implying $\| \hat \bJ - \bJ \|_{max} = o_p(1)$.
\end{proof}

\begin{lemma}\label{lem:jepsconverge}
    Under Conditions C1-C4, for any constant $r_1 \in (0, \frac{r}{8+r})$ and $r_2= \min(\frac{r}{32+4r}, \\ \frac{4+r}{32+4r}-\frac{r_1}{4})$, let $m_0=o(n^{r_2-c})$ for some arbitrarily small $c\in (0, \min\{r_1,r_2\} )$, then we have $\|  \hat \bJ_\xi - \bJ_\xi \|_{max} = o_p(m_0^{-2})$.
\end{lemma}
\begin{proof}
    Similar to the proof in Lemma \ref{lem:jconverge}, we let $\tilde J$ as defined in the proof of Lemma \ref{lem:jconverge}. Then with probability tending to 1,  $\|  \hat J_\xi -\tilde J_\xi \|_{max} =o_p(m_0^{-2}) $. By the definition of $\tilde \bJ$, the elements of $\hat \bJ$ are bounded by some constant times $n^{\frac{4}{8+r}}$ with probability tending to one by Lemma \ref{lem:hoeffding}. Thus it suffices to prove the lemma in this case. The Hoeffding's inequality for the centred U-statistic implies that for some constant $c_1>0$ and $c_2>0$,
    \begin{align*}
    \Pr\left( \| \tilde \bJ_\xi - \bJ_\xi \|_{max}> m_0^{-2} \right) 
    & \leq 2\vert \xi \vert^2\exp\left \{ -c_1 \cdot n^{\frac{r}{8+r}}  m_0^{-4}  \right \} \\
    & \leq 2\vert \xi \vert^2\exp\left \{ -c_1 \cdot n^{\frac{r}{8+r}}  n^{-4r_2+c_2}  \right \} \\
    & = o(1),
    \end{align*}
    since $\vert \xi \vert \le 3m_0$ and $4r_2\le \frac{r}{8+4r}$. Thus $\|  \tilde \bJ_\xi - \bJ_\xi \|_{max} = o_p(m_0^{-2})$, implying that $\|  \hat \bJ_\xi - \bJ_\xi \|_{max} = o_p(m_0^{-2})$.
\end{proof}

\begin{lemma} \label{lem:wnorm}
    Under Conditions C1-C4, for any constant $r_1 \in (0, \frac{r}{8+r})$ and $r_2= \min(\frac{r}{32+4r}, \\ \frac{4+r}{32+4r}-\frac{r_1}{4})$, let $m_0=o(n^{r_2-c})$ for some arbitrarily small $c\in (0, \min\{r_1,r_2\} )$. If $\Pr(\hatEps \subseteq \Eps) \to 1$ as $n\to\infty$, then $\| \hat \bw_\xi \|_1 = O_p(m_0)$ and $\hat \bw_{\Eps\setminus \xi}$ is a vector of ones with probability converging to 1.
\end{lemma}
\begin{proof}
    From the objective (\ref{eq:crit_true2}), if the lemma assumption holds, then $\Pr( \| \hat \bw \|_1 \leq \| \tilde \bw_\Eps \|_1  )\to 1$ as $n\to\infty$, where $\tilde \bw_\Eps = {\hat \bJ_{\Eps}}^{-1} \hat \bh_{\Eps}$ is a length $\vert \Eps \vert$ vector optimizing $\bw^\top \hat \bJ_\Eps \bw /2 - \bw^\top \hat \bh_\Eps$. Let $\bw_\Eps^\ast = \bJ_\Eps^{-1} \bh_\Eps$ be the minimizer of $\bw^\top \bJ_\Eps \bw /2 - \bw^\top \bh_\Eps$. Then $\hat \bJ_\Eps \tilde \bw_\Eps = \hat \bh_\Eps$ and $\bJ_\Eps \bw_\Eps^\ast = \bh_\Eps$. These two equations imply that $\bJ_\Eps(\tilde \bw_\Eps - \bw_\Eps^\ast) = \hat \bh_\Eps - \bh_\Eps + (\bJ_\Eps - \hat \bJ_\Eps)\tilde \bw_\Eps$. Moreover, by the definition of \ref{eq:hatJ}, $\hat \bJ_\Eps \bw = ( (\hat \bJ_\xi \bw_\xi)^\top, (\hat \bJ_{\Eps\setminus \xi} \bw_{\Eps\setminus \xi})^\top)^\top$ for any length $\vert \Eps \vert$ vector $\bw$. I.e., the equation can be separated into two, each corresponding to the subset $\xi$ and $\Eps\setminus \xi$. The same statement holds for $\bJ_\Eps$. By the KKT Condition \ref{eq:kkt}, this implies that $\hat \bw_{\Eps\setminus \xi} = \tilde \bw_{\Eps\setminus \xi} = \bw_{\Eps\setminus \xi}^\ast$ are vectors of ones, since $\hat \bJ_{\Eps\setminus \xi} $ and $\bJ_{\Eps\setminus \xi}$ are diagonal by definition when $\hatEps \subseteq \Eps$. Thus $\Pr( \| \hat \bw_\xi \|_1 \leq \| \tilde \bw_\xi \|_1  )\to 1$ as $n\to\infty$.
    
    For the equation corresponding to $\xi$, we have
    \begin{align*}
    \| \tilde \bw_\xi - \bw_\xi^\ast\|_\infty &\leq \| \bJ_\xi^{-1}\|_{max} \left( \| \hat \bh_\xi - \bh_\xi \|_1 + \| (\bJ_\xi - \hat \bJ_\xi)\tilde \bw_\xi  \|_1 \right)\\
    & \leq \| \bJ_\xi^{-1}\|_{max} \left( 3m_0\| \hat \bh_\xi - \bh_\xi \|_\infty + (3m_0)^2 \|\bJ_\xi - \hat \bJ_\xi\|_{max} \| \tilde \bw_\xi  \|_\infty \right)\\
    & = O(1) \left \{o_p(1) + o_p(1)\| \tilde \bw_\xi\|_\infty \right \},
    \end{align*}
    where the last equation is implied by Condition C2 and Lemma \ref{lem:jepsconverge}. Note that $\|\bw_\xi^\ast\|_\infty = O(1)$ under Condition C2. Thus
    $
        \| \tilde \bw_\xi \|_\infty - \| \bw_\xi^\ast \|_\infty \leq \| \tilde \bw_\xi - \bw_\xi^\ast\|_\infty = o_p(1) + o_p(1)\| \tilde \bw_\xi\|_\infty,
    $ 
    which implies $ \| \tilde \bw_\xi \|_\infty = O_p(1)$. Thus, $\Pr\{ \| \hat \bw_\xi \|_1 \leq \| \tilde \bw_\xi \|_1 \leq 3m_0\| \tilde \bw_\xi \|_\infty \le c m_0  \}\to 1$ as $n\to\infty$, for some constant $c>0$. Thus, $\| \hat \bw_\xi \|_1 = O_p(m_0)$.
\end{proof}

\begin{lemma}\label{lem:lambda}
    Under Conditions C1-C4, for any constant $r_1 \in (0, \frac{r}{8+r})$ and $r_2= \min(\frac{r}{32+4r},\\ \frac{4+r}{32+4r}-\frac{r_1}{4})$, let $m_0=o(n^{r_2-c})$ for some arbitrarily small $c\in (0, \min\{r_1,r_2\} )$. If $\Pr(\hatEps \subseteq \Eps) \to 1$ as $n\to\infty$, then $\hat \lambda = O(\gamma) O_p(m_0)$.
\end{lemma}
\begin{proof}
    By the assumption of the Lemma, $\Pr\{\hatEps \subseteq \Eps\} \to 1$ as $n\to\infty$. Thus, it suffices to consider the case where $\hatEps \subseteq \Eps$. By KKT condition (\ref{eq:kkt}) and the definition of $\hat\lambda$, the next pair $jk$ ($j<k$) about to enter $\hatEps$ satisfies $nS_{jk}^2 \leq \gamma (S_{jk}^2 + S_{jj}S_{kk})$ and
    $$
        \left |  S_{jk}^2 \left\{ \hat \bJ_{jk,\hatEps} \hat \bw_{\hatEps} - \hat h_{jk}  \right\} \right | = \frac{\hat \lambda}{n},
    $$
    where $h_{jk}$ is the element of $\bh$ corresponding to pair $jk$. Under Condition C1 and by Lemma \ref{lem:sconverge}, $0\leq S_{jk}^2 = O_p( \gamma/n) $. Moreover, with probability tending to one, $\hat\bJ_{jk,\hatEps} \hat\bw_{\hatEps} =  \hat\bJ_{jk,\hat\xi}\hat\bw_{\hatxi} + \hat\bJ_{jk,\hatEps\setminus\hat\xi}\hat\bw_{jk,\hatEps\setminus\hat\xi}\le \hat\bJ_{jk,\xi}\hat\bw_{\xi} + 2\| \hat \bJ \|_{max} = O(m_0)$ by Lemma \ref{lem:wnorm} and Lemma \ref{lem:jconverge}. Thus, by Condition C2, Lemma \ref{lem:jconverge} and Lemma \ref{lem:wnorm}, we have
    $
    \hat \lambda = O(\gamma) O_p(m_0).
    $
\end{proof}

\begin{proof}[\bf Proof of Theorem \ref{theory:fixedp}]
    We first show that $\Pr(\hatEps \subseteq \Eps) \to 1$. Under Condition C1, for all pairs $\{jk: j\leq k\}$, we have $S_{jk} \to \theta_{jk}$ in probability by the law of large number. Since $\theta_{jj}$ is bounded away from zero for all $j=1,\ldots,p$, we have $S_{jk}/(S_{jk}^2 + S_{jj}S_{kk})\to \theta_{jk}/(\theta_{jk}^2 + \theta_{jj}\theta_{kk})$ in probability. Thus there exists some constant $c>0$, such that $\Pr\{ S_{jk}/(S_{jk}^2 + S_{jj}S_{kk})  > c \ | \ \theta_{jk} = 0\} = o(1)$ and $\Pr\{ S_{jk}/(S_{jk}^2 + S_{jj}S_{kk})  < c \ | \ \theta_{jk} \neq 0\} = o(1)$. This implies that when $\gamma/n = c$, $\Pr(\hatEps \subseteq \Eps) \to 1$ as $n\to\infty$ by the definition of $\hat \lambda$ in (\ref{eq:sel_lambda}). 

    Next we show that $\hat \lambda/n=o_p(1)$. From the above discussion, it suffices to consider the case where $\hat\Eps \subseteq \Eps$. Note that $\Pr[ \max_{jk}\{ S_{jk}/(S_{jk}^2 + S_{jj}S_{kk}) \}  < c \ | \ \theta_{jk} \neq 0] = o(1)$. By the definition of $\hat \lambda$ and the KKT Condition, when $n$ is large, with probability arbitrarily close to 1, we have $S_{jk}^2 | \hat \bJ_{jk,\hatEps} \hat \bw_{\hatEps} - \hat h_{jk} | = \hat \lambda/n$, where $jk\in \setminus \Eps$ is some pair about to enter the set $\hatEps$ if $\hat \lambda$ decreases, and $\hat \bJ_{jk,\hatEps}$ is the row of $\hat \bJ$ corresponding to $jk$ and columns corresponding to $\hatEps$. With fixed $p$, convergence of $\hat \bJ$ to $\bJ$ is ensured by the law of large numbers. Thus, implied by the objective function in (\ref{eq:crit_true2}), $\| \hat \bw_{\hatEps} \|_1 \leq \| \tilde \bw_\Eps \|_1 \to \| \bw^\ast_\Eps \|_1 = O(1)$ in probability where $\tilde \bw_\Eps$ and $\bw_\Eps^\ast$ are the optimizers of $\bw^\top \hat \bJ_\Eps \bw/2 - \bw^\top \hat \bh_\Eps$ and $\bw^\top \bJ_\Eps \bw/2 - \bw^\top \bh_\Eps$, respectively. Since $\hat \bJ \to \bJ$ and $S_{jk}^2 \to 0$ in probability for $jk \in\setminus \Eps$, we have $\hat \lambda/n = S_{jk}^2| \hat \bJ_{jk,{\hatEps}}\hat \bw_{\hatEps} - \hat h_{jk} | = o_p(1)$

    Finally, we show that $\Pr( \Eps \subseteq \hatEps) \to 1$ as $n\to\infty$. From the above discussion, it suffices to show that every element of $\hat \bw_\Eps$ is nonzero, where $\hat \bw_\xi$ optimizes 
    $$ \frac{1}{2}\bw^\top \hat \bJ_\xi \bw/2 - \bw^\top \hat \bh_\xi + \frac{\hat \lambda}{n} \sum_{jk \in \xi} \frac{|w_{jk}|}{S_{jk}}^2. 
    $$
    Since $\hat \lambda/n =o_p(1)$ and $S_{jk} \to \theta_{jk}$ which is bounded away from zero for $jk \in \xi$, the above convex objective function of $\bw$ converges  point-wise to the strictly convex objective 
    $$ \frac{1}{2}\bw^\top \bJ_\xi \bw/2 - \bw^\top \bh_\xi. 
    $$
    Since $m_0$ is fixed, the optimizer of the former converges to the latter, i.e., $\hat \bw_\xi \to \bw_\xi^\ast$ in probability where $\bw_\xi^\ast$ is such that $\bJ_\xi \bw_\xi^\ast = \bh_\xi$. By Condition C3, all elements of $\bw_\xi^\ast$ are nonzero, implying the same for $\hat \bw_\xi$. Similar to the proof in Lemma \ref{lem:wnorm}, elements of $\hat \bw_{\setminus \xi}$ corresponding to the marginal scores are all equal one.
\end{proof}

\begin{proof}[\bf Proof of Theorem \ref{theory:false_positive}]
    Under Condition C1, for all $jk\in\setminus \Eps$, $E(S_{jk})=0$. Thus by the law of large number, $S_{jk} \to 0$, $S_{jj}\to\theta_{jj}$ and $S_{kk}\to\theta_{kk}$. Therefore, for any $jk\in\setminus \Eps$, $nS_{jk}^2/(S_{jk}^2 + S_{jj}S_{kk}) =o_p(n)$. However, by the definition of $\hat \lambda$ in (\ref{eq:sel_lambda}), for all $jk\in\hatEps$ and $j<k$, $nS_{jk}^2/(S_{jk}^2 + S_{jj}S_{kk}) > \gamma$. Thus for $\gamma^{-1} =o(n^{-1})$,
    $$
    \Pr\left ( \frac{nS_{jk}^2}{S_{jk}^2 + S_{jj}S_{kk}} > \gamma \ \bigg \vert \ jk \in \setminus \Eps\right ) \to 0,
    $$
    as $n\to \infty$. 
\end{proof}

\begin{proof}[\bf Proof of Theorem \ref{thm:consistency}]
    {\bf For part (i):} By definition of $\hat \lambda$ in (\ref{eq:sel_lambda}), for any $jk \in \hatEps$ and $j<k$, 
    $$\frac{nS^2_{jk} }{S_{jk}^2 + S_{jj}S_{kk}} > \gamma.$$ By Condition C1 and Lemma \ref{lem:sconverge}, there exists some constant $c>0$, $\Pr\{\min_{jk}(S_{jj}S_{kk})<c\}= o(1)$.
    Similarly, by Hoeffding's inequality for central U-statistic, $\Pr\{\max_{jk\in \setminus \Eps}(S_{jk}^2 )>c\cdot \gamma/n \}= o(1)$ for some $c>0$. Thus, there exist some constant $c_1>0$ for all $n$,
    \begin{align*}
        &\Pr\left\{ \max_{jk}\left(\frac{nS_{jk}^2}{S^2_{jk}+S_{jj}S_{kk}}\right)>\gamma \ \bigg| \ jk\in\setminus \Eps \right\}\\
        \leq &\Pr\left\{ \max_{jk}(nS_{jk}^2) > \gamma \min_{sl}(S_{ss}S_{ll}) \ \bigg| \ jk\in\setminus \Eps \right \}\\
        \leq & c_1 \cdot Pr\left\{ \max_{jk}(S_{jk}^2) > c \cdot \frac{\gamma}{n}   \ \bigg| \ jk\in\setminus \Eps \right \}\\
        \to & 0.
    \end{align*}
    That is $\Pr(\hatEps \cap \setminus \Eps = \emptyset) \to 1,$ as $n\to \infty$.

    {\bf For part (ii):} By part (i) of the theorem, $\Pr(\hatEps \subseteq \Eps) \to 1$ as $n\to \infty$. Thus it suffices to only consider the case when $\hatEps \subseteq \Eps$. By Lemma \ref{lem:wnorm}, $\hat w_{\Eps\setminus \hatxi}$ is a vector of ones. Moreover, by Condition C1 and Lemma \ref{lem:sconverge}, $S_{jk}$ are uniformly bounded and bounded away from zero, for all $jk\in\Eps$. Thus, the KKT condition (\ref{eq:kkt}) implies that
    $\| \hat \bJ_\Eps \hat \bw_\Eps - \hat \bh_\Eps \|_1 = O_p( \hat \lambda m_0/n)$. By Lemma \ref{lem:lambda} and theorem assumption, $\hat \lambda m_0/n = O_p(\gamma m_0^2/n)=o_p(1)$. Thus $\| \hat \bJ_\xi \hat \bw_\xi - \hat \bh_\xi \|_1 \le \| \hat \bJ_\Eps \hat \bw_\Eps - \hat \bh_\Eps \|_1 = o_p(1)$. Note that 
    \begin{align*}
        \| \hat \bJ_\xi \hat \bw_\xi -\hat \bh_\xi  \|_1 - \| \bJ_\xi \hat \bw_\xi - \bh_\xi \|_1 &\leq \left \| (\hat \bJ_\xi - \bJ_\xi) \hat \bw_\xi + (\bh_\xi - \hat \bh_\xi) \right \|_1\\
        & \leq m_0 \| \hat \bJ_\xi - \bJ_\xi \|_{max} \|\hat \bw_\xi \|_1 + m_0 \| \bh_\xi - \hat \bh_\xi \|_\infty \\
        & = o_p(1),
    \end{align*}
    where the last equation is implied by Lemma \ref{lem:jepsconverge} and Lemma \ref{lem:wnorm}. Thus $\| \bJ_\xi \hat \bw_\xi - \bh_\xi \|_1=o_p(1)$. Condition C3 implies that all elements of $\hat \bw_\xi$ are nonzero. That is $\Eps \subseteq \hatEps$.
\end{proof}

\end{appendix}

\bibliographystyle{abbrvnat} 
\bibliography{bibliography}

\end{document}